\documentclass[12pt, draftclsnofoot, onecolumn]{IEEEtran}
\usepackage{cite}
\usepackage{amsmath,amssymb,amsfonts}
\usepackage{algorithmic}
\usepackage{graphicx}
\usepackage{textcomp}
\usepackage{times}
\usepackage{bm}
\usepackage{multirow}
\usepackage{enumerate}
\usepackage{graphicx}
\usepackage{ifpdf} 
\usepackage[format=plain,labelfont=bf,up,textfont=it,up]{caption}
\usepackage{subfigure} 
\usepackage{amsthm}
\usepackage{multirow}
\usepackage{algorithmic}
\usepackage{algorithm}
\usepackage{caption}
\DeclareCaptionLabelSeparator{twospace}{.\ ~}
\newtheorem{theorem}{Theorem}

\newtheorem{remark}{Remark}

\def\BibTeX{{\rm B\kern-.05em{\sc i\kern-.025em b}\kern-.08em
		T\kern-.1667em\lower.7ex\hbox{E}\kern-.125emX}}
\begin{document}
 
\title{Computation Offloading in the Untrusted MEC-aided Mobile Blockchain IoT System }
\author{Yiping Zuo,
        Shi Jin,~\IEEEmembership{Senior Member,~IEEE,}
        and~Shengli Zhang,~\IEEEmembership{Senior Member,~IEEE}
\thanks{Y. Zuo and S. Jin are with the National Mobile Communications Research Laboratory, Southeast University, Nanjing 210096, China (e-mail: zuoyiping@seu.edu.cn; jinshi@seu.edu.cn).} 
\thanks{S. Zhang is the College of Information Engineering,
	Shenzhen University, Shenzhen 518060, China (e-mail: zsl@szu.edu.cn) }
}
\maketitle
\begin{abstract}
Deploying mobile edge computing (MEC) server in the mobile blockchain-enabled Internet of things (IoT) system is a promising approach to improve the system performance, however, it imposes a significant challenge on the trust of MEC server. To address this problem, we first propose an untrusted MEC proof of work scheme in mobile blockchain network where plenty of nonce hash computing demands can be offloaded to MEC server. Then, we design a nonce ordering algorithm for this scheme to provide fairer computing resource allocation for all mobile IoT devices/users. Specifically, we formulate the user's nonce selection strategy as a non-cooperative game, where the utilities of individual user are maximized in the untrusted MEC-aided mobile blockchain network. We also prove the existence of Nash equilibrium and analyze that the cooperation behavior is unsuitable for the blockchain-enabled IoT devices by using the repeated game. Finally, we design the blockchain's difficulty adjustment mechanism to ensure stable block times during a long period of time. 

\end{abstract}
\IEEEpeerreviewmaketitle
\section{Introduction}
As an emerging frontier technology, blockchain has received increasing attention in recent years. The success of the blockchain concept is ultimately connected with Bitcoin \cite{nakamoto2008bitcoin}. Blockchain technology is not only limited in Bitcoin, but also used in many communication areas, such as Internet of things (IoT)\cite{dai2019blockchain,cao2019when,conoscenti2016blockchain,banerjee2018blockchain,danzi2018analysis,danzi2019delay}, edge computing\cite{xiong2018when,xiong2018cloud,jiao2018social,jiao2019auction,luong2018optimal}, spectrum sharing \cite{Bayhan2018Spass,Kotobi2017blockchain} and interference management\cite{Gamal2019ASingle}. The security and reliability of blockchain are mainly determined by the distributed consensus mechanisms, for example, proof of work (PoW) mining mechanism (see Section \ref{PoW} for details) for Bitcoin system\cite{garay2015bitcoin}, proof of stake for Ethereum \cite{bentov2016cryptocurrencies} and so on. In essence, blockchain is a tamper-proofing and distributed database, which records transactional data without the requisition of a trusted authority or central server in the peer-to-peer (P2P) network. 

Although blockchain has many advantages, heavy equipments and fixed access points have become shortcomings of the traditional blockchain network, which can limit its further development. To tackle this problem, mobile blockchain with large-scale IoT devices \cite{cao2019when} would be a better choice in the future. Nevertheless, IoT devices have poor privacy, security and scalability. Fortunately, the emergence of blockchain technology creates opportunities to overcome the aforementioned drawbacks. The research on mobile IoT devices deployed in the mobile blockchain network is emerging. Recently, some researches on the combination of blockchain and IoT technologies have been published \cite{conoscenti2016blockchain,banerjee2018blockchain,danzi2018analysis,danzi2019delay}.
Specifically, authors in \cite{banerjee2018blockchain} had a systematic survey on IoT security and proposed that the frontier blockchain technology has enormous potentials as the solution scheme. \cite{danzi2018analysis} has investigated synchronization protocols with different communication costs and security levels, which synchronize mobile IoT users with the blockchain network. Blockchain technology provides a distributed architecture for coordinating IoT devices, whereas it is not feasible to store the full blockchain ledger for low-power and memory-constrained IoT devices. Thus authors have proposed a lightweight software scheme, where IoT devices only download useful data structures in \cite{danzi2019delay}.

The above researches improved the system performance to some extent. Nevertheless, mobile blockchain network is constrained by consuming too much hash computing, storage, and energy resources during the mining process on mobile IoT devices. It is unable to satisfy some characteristics of IoT devices with relatively low computing ability, low power consumption, and scattered low-bandwidth wireless connections.  
Instead, mobile edge computing (MEC) is a network architecture concept that implements cloud computing capabilities and an IT service environment at any network edge \cite{mach2017mobile,mao2017a,yu2018a}. So MEC can supply these resources for mobile blockchain-enabled IoT devices. Deploying MEC server in the mobile blockchain network is a feasible way to handle the low computing power dilemma \cite{xiong2018when,xiong2018cloud,jiao2018social,jiao2019auction,luong2018optimal}. In particular, the PoW mining process focuses on finding a correct nonce for solving PoW puzzle. While it is unrealistic for mobile IoT users to continuously run such a computation-intensive mining task, which requires high computing ability and consumes substantial energy and time. Due to the prominent characteristics of MEC server, such as low latency, high mobility, and wide regional distribution, we consider to offload the blockchain mining tasks to MEC server.     
                 
Currently, there is substantial interest in investigating computation offloading in the mobile blockchain network with MEC server. For example, prior works in \cite{xiong2018when,xiong2018cloud} have formulated the interaction between MEC server and users as the two-stage Stackelberg game in the mobile blockchain network. The Stackelberg equilibrium ensures that MEC server maximizes its revenue. Alternatively, it is also significant to maximize the profit of mobile users.
Auction schemes in \cite{jiao2018social,jiao2019auction,luong2018optimal} are the feasible solutions to accomplish this goal. The work in \cite{jiao2018social} has investigated an auction-based MEC system model to maximize the revenue of mobile users. Essentially, the auction of the edge computing service is an optimization problem, which tackle the biggest profit of users. This optimization problem is a nonlinear programming problem with linear constraints, which is NP-hard. Authors in \cite{jiao2019auction} have applied the approximate fractional relaxation and local search (FRLS) method to ensure a lower bound. Whereas MEC server cannot achieve optimal auction by using traditional single-item auctions. An optimal single-item auction for the edge computing resource allocation was developed via a multi-layer neural network architecture based on an analytical solution of the optimal auction in \cite{luong2018optimal}. Another line of research takes content cashing and different offloading modes into consideration \cite{liu2018computation,wu2018optimal,liu2019distributed}, because PoW consensus mechanism of blockchain also requires to spend lots of storage resources in the mining process. To tackle this problem, the work in \cite{liu2018computation} has considered that computationally difficult mining task can be offload to nearby edge computing nodes and transactional content of blocks can be cached in the MEC server. This paper has proposed two offload modes, offloading to the nearby access point (AP) or a group of nearby D2D devices. In addition, the research in \cite{wu2018optimal} has also considered the advanced multi-access MEC server in the mobile blockchain network. 

In these previous researches, many works have investigated the resource allocation of MEC server in a mobile blockchain network with the assumption that MEC is a trusted server. However, in this case, MEC server will have an opportunity to profit from the usage of information of IoT users. It is likely that MEC server will allocate more computing resources to some selfish users so as to obtain more revenue. Generally speaking, if MEC is a trusted server, it may lead to unfair computing resource allocation for mobile users and also increase the possibility of malicious collusion between MEC server and selfish users. To the best of authors' knowledge, few works focus on untrusted MEC servers in the mobile blockchain network. Moreover, the previous works treated the mining task as a general computing demand, which really reflect the characteristics of the blockchain mining task. In this paper, different from previous works, we consider an untrusted MEC PoW scheme and specially design a nonce ordering algorithm for the scheme. In the untrusted MEC PoW framework, user's nonce selection strategy is formulated as a non-cooperative game, where the utilities of individual user are maximized. Meanwhile, we design a blockchain's difficulty adjustment mechanism for keeping network stable.  

In this paper, we intend to obtain insights into users' nonce selection strategies and blockchain's difficulty adjustment mechanism in the proposed untrusted MEC-aided mobile blockchain network. The main contributions of this paper are listed as follows.
\begin{itemize}
	\item[1)] We propose an untrusted MEC PoW scheme and a nonce ordering algorithm to provide fairer nonce computing resource allocation for all mobile IoT users. For any mobile user, such ordering algorithm guarantees that the probability of successfully mining a new block corresponding to other users is proportional to its nonce length.
	
	\item[2)] Based on \textit{Theorem} \ref{SuccessfulProbability}, we formulate the user's nonce selection strategy as a non-cooperative game, in which the utilities of individual user are maximized. We analyze the Nash equilibrium (NE) existence in this game and propose an alternating optimization algorithm of users' nonce selection. Furthermore, we drive the NE in analytical expression which only needs to know partial information of other users and analyze that cooperation is unsuitable for the blockchain-enabled IoT devices in the repeated game.
	
	\item[3)] We design the blockchain's difficulty adjustment mechanism to keep block times stable during a long period of time. The mean and variance of rounds for mining $G$ blocks are derived. Moreover, the adjustment rule of difficulty factor is to minimize the distance between the average rounds for $G$ blocks and the predefined threshold.
	
	\item[4)] Simulation results verify the effectiveness of our proposed nonce ordering algorithm under the untrusted MEC PoW scheme. Then we demonstrate that the expected optimal user's nonce selection strategy is influenced by multiple system parameters jointly, such as the block size, the fixed reward for successfully, the transaction fee rate, etc. In addition, the design of blockchain's difficulty adjustment mechanism achieves a good performance over a long period of time. 
\end{itemize}

The rest of the paper is organized as follows. Section \ref{sec:systemmodel} introduces the PoW mining process and system model of MEC-aided mobile blockchain network. Section \ref{sec:untrustedMECPoWscheme} represents our proposed untrusted MEC PoW scheme and nonce ordering algorithm. Section \ref{SystemFormulationandAnalysis} formulates user's selection strategy problem, designs blockchain's difficulty adjustment mechanism and affords some theoretical analysis. Simulations are presented in Section \ref{sec:SimulationResults} to confirm the analytical results. At last, we conclude the main results of the work in Section \ref{sec:conclu}.

\section{Preliminaries and System Model}\label{sec:systemmodel}
In this section, we will introduce the background on PoW mining mechanism in the Bitcoin network and MEC system model of edge computing resource allocation for mobile blockchain network.
\subsection{PoW Mining Mechanism}\label{PoW}
As a matter of fact, PoW mining mechanism is to prevent low-computing malicious entities from publishing blocks arbitrarily. To obtain the billing power, nodes in the network have to construct a block in advance and select some unconfirmed transaction records into a new block. Then, each node solves the PoW puzzle, making the hash value of the block header smaller than the blockchain's target value by changing nonces randomly. The input of the hash function is the block header and the output is a 256-bit hash value. This process of solving PoW puzzle is called mining. Once the PoW puzzle is resolved, the newly mined block will be immediately announced to the whole blockchain network. Meanwhile, the other nodes receive this information and execute a validating process to decide whether to approve and add a newly generated block to the blockchain or not. Each extension of this block is equivalent to an additional confirmation of the transaction in the block. If getting 6 confirmations, the block is approved by the whole network and encapsulated in the historical block. The miner which successfully mines a new block will achieve a certain number of reward, including a fixed bonus and a variable transaction fee, as an incentive of mining. 

The Bitcoin system will control the completion time of mining a new block in about 10 minutes. If the block time is less than 10 minutes, the Bitcoin system will automatically increase the difficulty value and the number of 0 at the beginning of the target hash value. Oppositely, if the block time is higher than 10 minutes, the number of 0 at the beginning of the target hash value is appropriately reduced to obtain a lower difficulty value. To simplify the explanation, let $H\left(  \cdot  \right)$ denote the hash function which uses a cryptographic algorithm to generate a short summary from any size of data. $X$ denotes block header information, such as version, the hash value of the previous block, Merkle root and so on. Given an adjustable hardness condition parameter $h$, the process of PoW puzzle solution aims to search a correct nonce to be included in the block. The target hash value of block header $bh$ which concatenates $X$ and nonce is smaller than a target value $V(h)$:
\begin{equation}\label{difficulty}
bh = H\left( {X||nonce} \right) \le V\left( h \right),
\end{equation}
where we have $V(h) = {2^{L - h}} = \frac{2^{L}}{{D\left( h \right)}}$, $L$ denotes the fixed length of bits, determining the searching space of the hash function, i.e., all $nonce \in \left[ {0,{2^L-1}} \right]$, and  $D(h)$ is blockchain's difficulty value.
\subsection{MEC System Model}
The mobile blockchain network is constrained because the mining process demands too much computing, storage, and energy resource on mobile IoT devices. Instead, MEC server can supply these resources for mobile blockchain-enabled IoT devices. We only consider computing resources in this paper. Fig.~\ref{fig:edgecomputingbc} depicts the system model of MEC-aided mobile blockchain network, which includes single MEC and $N$ mobile IoT devices/users running blockchain application denoted as ${\cal N} = \left\{ {1,...,N} \right\}$. Due to the computing limitation on mobile IoT devices, users want to offload the task of solving the PoW puzzle to MEC server. Besides, all mobile IoT users send requests to MEC server in a time slot.  
\begin{figure} [!htb]
	\centering
	\includegraphics[width=0.6\linewidth]{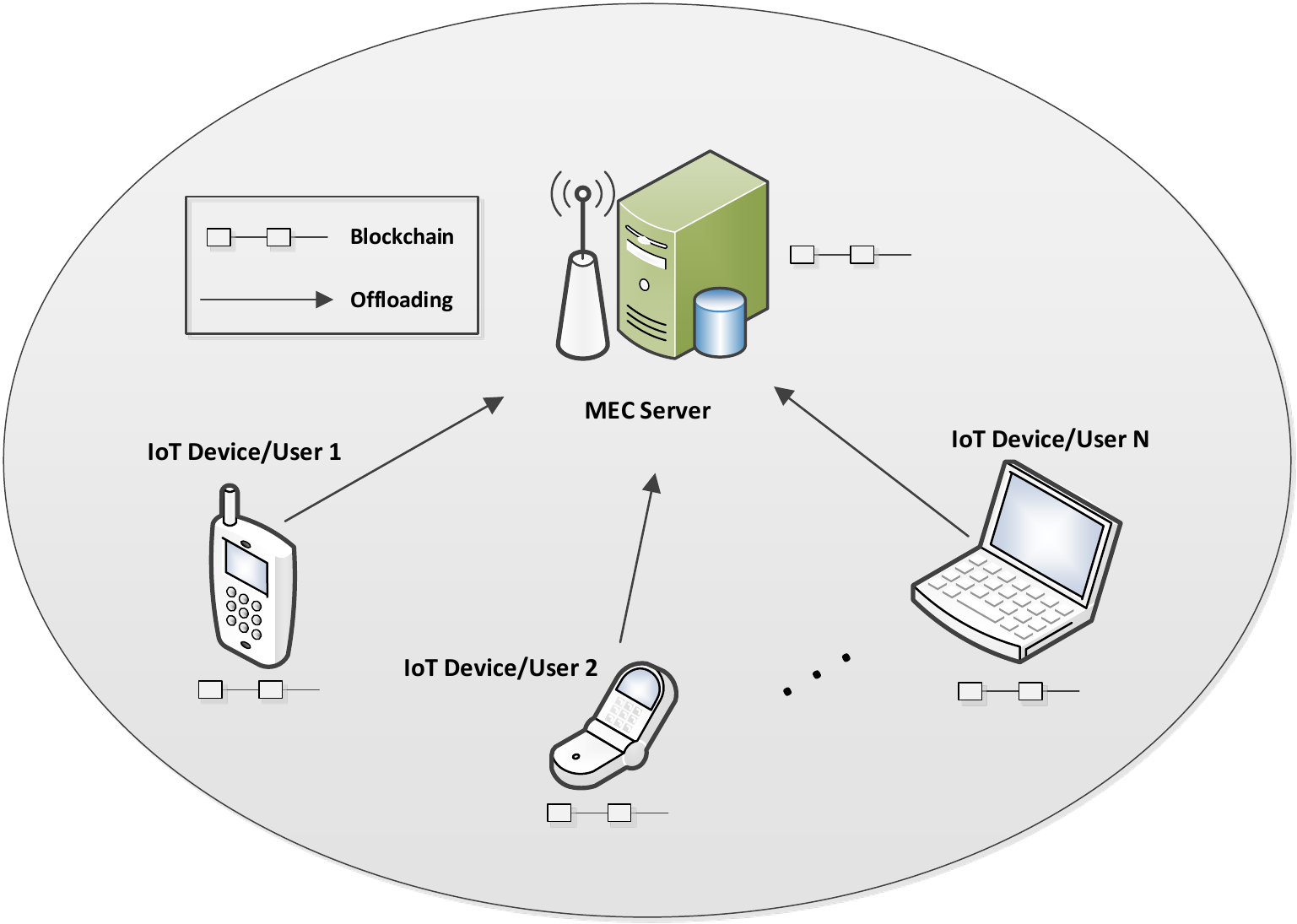}
	\caption{Edge computing in the mobile blockchain network.}
	\label{fig:edgecomputingbc}
\end{figure}
Previous researches on the MEC server in a mobile blockchain network have the assumption that MEC is a trusted server. If MEC is a trusted server, it may lead to unfair computing resource allocation for mobile users and increase the possibility of malicious collusion between MEC server and selfish users. Consequently, we consider an untrusted MEC server in this paper and focus on designing a reasonable allocation mechanism to make computing resource allocation fairer for all users. 

On the other hand, each user independently selects nonces for hash computing because the block content of each user is different. Even though each user chooses the same nonce, the corresponding hash is different (if the SHA-256 hash function is used, then all $nonce \in \left[ {0,{2^{32}}} \right]$). For mobile IoT devices, the computation-intensive mining demand is too high. Given the cost of renting MEC's service and mining time, the final user's revenue can be reduced or even negative. Thus we assume that each user only chooses some nonces to do hash computing. 

With the combination of above analysis, we put forward an untrusted MEC PoW scheme and design a nonce ordering algorithm for this mobile blockchain network. The detailed process of this new scheme is described in the following Section \ref{sec:untrustedMECPoWscheme}. 
\section{Untrusted MEC PoW scheme} \label{sec:untrustedMECPoWscheme}
In this section, we focus on the untrusted MEC PoW scheme in the mobile blockchain network. We design a nonce ordering algorithm for this scheme to achieve fairer hash computation resource of MEC server for all mobile IoT users.
\subsection{General Process of the Untrusted MEC PoW}
Selected nonces of each user are open and transparent in the whole network so that it can reduce the cheating possibility by the untrusted MEC server. Then MEC server provides hash computing service to these nonces. Mobile IoT user $i$ randomly selects a part of nonces as own computing demand by vector ${m}_i$ for $i=1,2,...,N$ and the nonce number/length of ${m}_i$ is denoted as ${M_i} = \left| {{m_i}} \right|$. Without loss of generality, nonces in the sequence $m_i$ are ordered from small to large by ${m_i} = \left\{{m_1^i,m_2^i,...,m_{{M_i}}^i} \right\}$, where $m_f^i$ is the element of $m_i$ for $f=1,2,...,M_i$. The MEC server, i.e., the seller, sells the computing services, and the users, i.e., the buyers, access and consume this service from the nearby MEC. All users submit their nonce sequence demand profile ${\bf{m}} = \left\{ {{m_1},{m_2},...,{m_N}} \right\}$ to MEC server. Next, we represent the general process of the untrusted MEC PoW scheme as follows.

After having received users' computation demands, the MEC server provides the hash computing service and achieves payment from all mobile IoT users. We introduce the detailed interaction process between MEC server and users in the untrusted MEC PoW scheme as follows.

\begin{itemize}
	\item \textbf{Step 1:} Users select nonces. Each mobile user randomly selects some nonces as its own computing demand. 
	
	\item \textbf{Step 2:} Submit tasks to MEC. All mobile users submit their mining tasks, i.e., the nonce hash computing demand profile $\bf{m}$ to MEC server.
	
	\item \textbf{Step 3:} MEC accepts the tasks. The MEC server accepts mining tasks for the users who arrive in a time slot and those users arriving after this time slot are not accepted. We assume that $N$ users submit the nonce sequences to the MEC server for hash computing in a time slot. 
	
	\item \textbf{Step 4:} Nonce ordering. Since MEC is an untrusted server, it is necessary to order nonces before providing mobile users with the computing service, so that MEC provides the computing service to users as fair and transparent as possible. The ordering algorithm adopted in this paper is introduced below (see Section \ref{subsection:nonceordering} ).
	
	\item \textbf{Step 5:} MEC provides services. The MEC server provides the computing service and receives payment from users. Once a user has coped with the PoW puzzle, MEC server will stop all tasks immediately and announce the result to all mobile IoT users. Then a new round of computing resource allocation will begin. 	
\end{itemize}
\subsection{Nonce Ordering} \label{subsection:nonceordering}
The following work highlights that how to order users' nonces on the MEC server for providing users computing services fairly, that is, how to map multiple nonce sequences into a sequence. We introduce this ordering algorithm for the untrusted MEC PoW scheme as follow:	
\begin{figure*}[!t]
	\centering
	\includegraphics[width=1.0\linewidth]{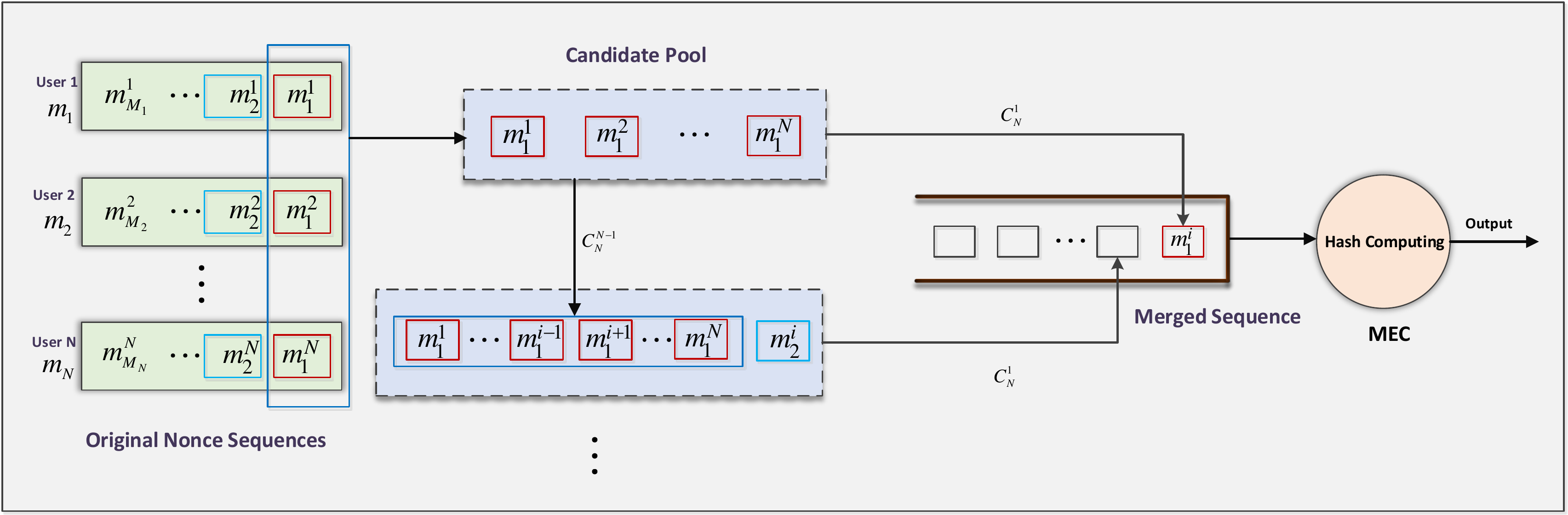} 
	\caption{The nonce ordering for the untrusted MEC server. $N$ original nonce sequences can be mapped into a merged sequence according to an ordering algorithm, which selects one nonce in each round of nonce ordering, then the remaining nonces in the previous round are treated as elements of the next round into candidate pool. Finally, the untrusted MEC server provides hash computing service for this merged nonce sequence.	
	}	
	\label{fig:noncequeuing}		
\end{figure*}
Firstly, each user's nonces have a local order from small to large, i.e.,  $m_1^i < m_2^i < ... < m_{{M_i}}^i$ for $i=1,2,...,N$. The original $N$ sequences are denoted $m_1,m_2,...,m_N$. We use another $N$ sequences, $w_1,w_2,...,w_N$ to denote the nonces selected from the corresponding original sequences to the merged sequence. Initially, the sequence $w_i$ is empty for $i=1,2,..,N$. 
As showed in Fig.~\ref{fig:noncequeuing}, the untrusted MEC server selects the top nonce in each user's nonce sequence as the alternative set into candidate pool. We intend to design a fairer nonce ordering algorithm to select a nonce in each round. Then the remaining nonces in the previous round are taken as elements of the next round into candidate pool. We desire that $N$ nonce sequences are merged into one sequence as shown in Fig.~\ref{fig:noncequeuing}. The objective of our ordering algorithm is that for any position in the merged sequence, the number of selected nonces of each original sequence is proportional to its nonce length. So the target probability of being served for user $i$ is
\begin{equation}\label{probability_p1}
{p_i} = \frac{{{M_i}}}{{\sum\limits_{j \in \mathcal{N}} {{M_j}} }},\;\text{for}\;i = 1,2,...,N,
\end{equation}
and the target probability mass is ${{P}}{\text{ = }}\left\{ {{p_1},{p_2},...,{p_N}} \right\}$.
Suppose the length of sequence $m_i$ is $l_i$ and the length of sequence $w_i$ is $k_i$. 
In each round of nonce ordering, a nonce of each original sequence is selected to the candidate pool as shown in Fig.~\ref{fig:noncequeuing}, and then a nonce is randomly selected among these $N$ nonces to the merged sequence. The index of the round is denoted as $n$ and the lengths of sequence $m_i$ and sequence $w_i$ are ${l_i}\left( n \right) = {l_i}\left( {n - 1} \right) - 1$ and ${k_i}\left( n \right) = {k_i}\left( {n - 1} \right) + 1$ respectively.  We calculate the expected winning probability of being served for user $i$, suppose it to be ${q_i} = {k_i} \textfractionsolidus \sum\nolimits_{j \in N} {{k_j}}$ and the expected winning probability mass is ${{Q}}{\text{ = }}\left\{ {{q_1},{q_2},...,{q_N}} \right\}$. In each round of nonce ordering, the top nonce of the $i$-th sequence $m_i$ is selected, so that there will be the distance between probability mass $P$ and $Q_i$, where $Q_i$ the expected winning experience probability distribution of user $i$ for $i=1,2,...,N$. In this case, the two probability distributions can be rewritten by  
\begin{equation}\label{probability_p2}
P = \left( {\frac{{{M_1}}}{{\sum\limits_{j \in \mathcal{N}} {{M_j}} }},\frac{{{M_2}}}{{\sum\limits_{j \in \mathcal{N}} {{M_j}} }},...,\frac{{{M_N}}}{{\sum\limits_{j \in \mathcal{N}} {{M_j}} }}} \right), \;\text{for}\;i = 1,2,...,N.
\end{equation}
\begin{equation}\label{probability_Q2}
{Q_i} = \left( {\frac{{{k_1}}}{{\sum\limits_{j \in \mathcal{N}} {{k_j}}  + 1}},...,\frac{{{k_i} + 1}}{{\sum\limits_{j \in \mathcal{N}} {{k_j}}  + 1}},...,\frac{{{k_N}}}{{\sum\limits_{j \in \mathcal{N}} {{k_j}}  + 1}}} \right),\;\text{for}\;i = 1,2,...,N.
\end{equation}
We will take a look at the way of comparing two probability distributions called Kullback-Leibler(KL) divergence. Next, we offer the definition of KL divergence as follows.\\
\textbf{Definition 1.} If we have two separate probability distributions $Q(M)$ and $P(M)$ over the same random variable $M$, we can measure the difference between the two distributions using the KL divergence:
\begin{equation}\label{KL1}
\begin{gathered}
{D_{KL}}\left( {Q\left\| P \right.} \right) \hfill \\
= {E_{M \sim Q}}\left[ {\log \frac{{Q\left( M \right)}}{{P\left( M \right)}}} \right] \hfill \\
= {E_{M \sim Q}}\left[ {\log Q\left( M \right) - \log P\left( M \right)} \right]. \hfill \\ 
\end{gathered} 
\end{equation}
We have a target probability distribution $P$ and wish to approximate it with another expect winning probability distribution $Q_i$, so that the probability of selecting each nonce for all users is fairer for $i=1,2,...,N$. The KL divergence between the two probability distributions $P$ and $Q_i$ is as follows,
\begin{equation}\label{KL2}
\begin{gathered}
{D_{KL}}\left( {{Q_i}\left\| P \right.} \right) \hfill \\
= \sum\limits_{i = 1}^N {{q_i}\log \frac{{{q_i}}}{{{p_i}}}}  \hfill \\
= \log \frac{{\sum\limits_{j \in \mathcal{N}} {{M_j}} }}{{\sum\limits_{j \in \mathcal{N}} {{k_j} + 1} }} - \left( {\sum\limits_{j = 1,j \ne i}^N {\left( {\frac{{{k_j}}}{{\sum\limits_{j \in \mathcal{N}} {{k_j} + 1} }}\log \frac{{{M_j}}}{{{k_j}}}} \right)} } \right. \left. { + \frac{{{k_i} + 1}}{{\sum\limits_{j \in \mathcal{N}} {{k_j} + 1} }}\log \frac{{{M_i}}}{{{k_i} + 1}}} \right). \hfill \\ 
\end{gathered} 
\end{equation}
By convention, in the context of information theory, we treat these expressions as $\lim _{k_i \rightarrow 0} k_i \log k_i = 0$, for $i=1,2,...,N$. We have the choice of minimizing ${D_{KL}}\left( {{Q_i}\left\| P \right.} \right)$ by 
\begin{equation}\label{optimal_i}
{i^*} = \arg \mathop {\min }\limits_i {D_{KL}}\left( {{Q_i}||{P}} \right), \;\text{for}\;\; i=1,2,...,N.
\end{equation}
So the $i^*$-th user's top nonce is selected to the merged sequence\footnote{If there are more than one nonce choices, we select one randomly in this case.}. We use the similar selecting rule as (\ref{optimal_i}) to select nonces to enter into the merged sequence, then multiple users' nonce sequences can be mapped into a sequence. Eventually, combining the above analysis part, the nonce ordering algorithm for the untrusted MEC PoW scheme can be obtained, as summarized in Algorithm \ref{orderingalgorithm1}.

\begin{algorithm}[!t]
	\caption{Ordering Algorithm for the Untrusted MEC PoW Scheme} \label{orderingalgorithm1}
	\setcounter{algorithm}{1}
	\begin{algorithmic}[1]
		\STATE {\textbf{Initialization:} input $N$ original nonce sequences $m_1$, $m_2,...$, $m_N$.}
		\REPEAT
		\FOR{$i=1$ to $N$}
		\STATE{Select one nonce from the $m_i$ sequence, calculate $Q_i$ by the equation (\ref{probability_Q2}).} \STATE{Then calculate the KL divergence between $P$ and $Q_i$ using the equation (\ref{KL2}).}
		\ENDFOR
		\STATE {In each round, select one nonce into a merged sequence by the rule (\ref{optimal_i}).}
		\UNTIL{all nonces of $N$ original sequences are mapped into a merged sequence.}
	\end{algorithmic}
\end{algorithm}
To analyze the performance of this untrusted MEC-aided mobile blockchain network system, we want to provide a mining probability for each user corresponding to other users in the following theorem under the above-proposed ordering algorithm. 
\begin{theorem}\label{SuccessfulProbability}
	For any mobile user $i$, the probability of successfully mining a new block corresponding to other users is nearest to this target probability: 
	\begin{equation}\label{UserProbability}
	{P_i^m} = \frac{{{M_i}}}{{\sum\limits_{i \in \mathcal{N}} {{M_i}} }},\;\;\text{for}\;i = 1,2,...,N.
	\end{equation}
\end{theorem}
\begin{proof}\label{ProofSuccessfulProbability}
	We prove this theorem by using reduction to absurdity. We assume that the $i$-th mobile user's probability of successfully mining a new block corresponding to other users, i.e., the probability of the untrusted MEC server providing nonce hash computing service in the merged sequence, is not nearest to the target probability $p_i$ under above proposed ordering algorithm.
	
	Under this assumption, it is easy to obtain that the $i$-th mobile user's nonces in the merged sequence obviously cannot satisfy the nonce selecting rule (\ref{optimal_i}) of our proposed ordering algorithm. Meanwhile, the selecting rule (\ref{optimal_i}) requires that each selected nonce has the smallest KL divergence between the expected winning probability and target probability distribution, namely, the expected winning probability approximates the target probability. When the total nonce length approaches infinity, the user's successful mining probability is almost equal to the target probability. Nevertheless, it directly contradicts our current assumption that the $i$-th mobile user's probability of successfully mining a new block is not nearest to the target probability. Hence, the statement in this theorem is proved.
\end{proof}

\begin{remark}\label{PropertiesTheorem1}
	(Property of Theorem \ref{SuccessfulProbability}): Our proposed nonce ordering algorithm in the untrusted MEC PoW scheme achieves better fairness of computing resource allocation among all mobile IoT users. For any position in the merged sequence of the untrusted MEC server, the number of selected nonces of each original sequence is proportional to its nonce length. 
\end{remark}

Previous researches have the assumption that MEC is a trusted server, resulting in many problems. For example, when multiple users submit mining task requests to a trusted MEC, users will inform the MEC server of all their information. This will offer MEC an opportunity to profit from the usage of this information, and it is likely that the MEC server will allocate more computing resources to some selfish users in order to obtain more revenue. Generally speaking, if the MEC is a trusted server, it may lead to unfair computing resource allocation for mobile users and also increase the possibility of malicious collusion between MEC server and selfish users. Accordingly, our proposed nonce ordering algorithm in the untrusted MEC PoW scheme is more reasonable and implement fairer computing resource distribution of MEC server among all IoT users in the mobile blockchain network.
\section{System Formulation and Analysis}\label{SystemFormulationandAnalysis}
In this section, we first present the non-cooperative game formulation for the mobile users' nonce selection strategy in the untrusted MEC PoW scheme. Then we analyze the NE existence of this game and propose an alternating optimization algorithm for users' nonce selection. We also derive the NE in analytical expression which only needs to know partial information of other users. Furthermore, we analyze that cooperation behavior is unsuitable for the blockchain-enabled IoT devices by using the repeated game. In the end, we design the difficulty adjustment process for the untrusted MEC-aided mobile blockchain network in detail.
\subsection{Non-cooperative Game Formulation}
In the PoW mining process, users compete to be first to solve the PoW puzzle with right nonce and broadcast the block to reach agreement in the whole blockchain network. We describe the size of transaction (TX) of each user by ${\bf{s}} = \left( {{s_1},...,{s_N}} \right)$. The first user which successfully mines a block and achieves agreement can get a reward $F$. The reward is composed of a fixed bonus $B$ for mining a new block and a flexible transaction fee determined by the size of its collected transactions ${\bf{s}}$ and the transaction fee rate $r$. Then mobile user $i$'s expected reward $F_i$ can be expressed by
\begin{equation}\label{userreward}
{F_i} = \left( {B + r{s_i}} \right)P_i^m,\;\text{for}\;i=1,2,...,N,
\end{equation}
where ${P_i^m}$ is the probability that user $i$ receives the reward by contributing a block to the blockchain corresponding to other users. From the nonce ordering algorithm above mentioned, the probability of mining a new block is proportional to user $i$'s nonce length, i.e., $P_i^m = {M_i} \textfractionsolidus \sum\nolimits_{j \in N} {{M_j}}$, for $i=1,2,...,N$. 

The untrusted MEC-aided blockchain network system discussed in this paper has only one MEC server. There is no fork and the orphaning probability will not be considered. While equation (\ref{userreward}) does not reflect the effect of blockchain's difficulty on the mining process. Even though all users' nonces are provided hash computing service by MEC server, the block will not be mined successfully. 
Nonce hash computing is a memoryless searching process, and the searching probability is only related to the difficulty value $D(h)$, regardless of the size of this searching space.
For a given difficulty value $D(h)$, each nonce hash computing is i.i.d Bernoulli trial with a successful probability
\begin{equation}\label{PD}
{P_D} = \frac{1}{{D\left( h \right)}}=2^{-h}.
\end{equation}
With this effect in mind, our equation for the $i$-th user's expected revenues gets discounted by the chances of blockchain's difficulty value, $P_{D}$, becoming
\begin{equation}\label{finalreward}
{F_i} = \left( {B + r{s_i}} \right) {2^{-h}} \frac{{{M_i}}}{{\sum\limits_{j \in \mathcal{N}} {{M_j}} }},\;\text{for}\;i=1,2,...,N.
\end{equation}

Once receiving computing resource demands from all mobile users, MEC server will provide hash computing service for mobile users' nonce sequence, and finally finds the correct nonce for mining. Here, we assume the price of each hash computing for the MEC is fixed by $c$. The users compete with each other to maximize their own utility by choosing their individual nonce computing demand, which forms the non-cooperative game $\mathcal{G} = \left( N, {\mathbf{M}}, {u_i} \left(  \cdot  \right) \right)$ as follows:\\
\textbf{Player:} the $N$ mobile users;\\
\textbf{Actions:} each player selects some nonces and corresponding $i$-th user's nonce length denoted as $M_i$;\\
\textbf{Utility function:} the utility function $u_i$ is denoted as the $i$-th user's revenue.\\
The revenue of $i$-th user is that the achieved rewards minus computing service cost of MEC server. Given the price of once hash computing $c$, the user $i$ determines its nonce hash computing service demand by maximizing the expected utility function which is given as:
\begin{equation}\label{user profit}
{u_i}\left( {{M_i}} \right) = \left( {B + r{s_i}} \right){2^{-h}}\frac{{{M_i}}}{{\sum\limits_{j \in \mathcal{N}} {{M_j}} }} - c{M_i},\;\text{for}\;i = 1,2,...,N.
\end{equation}
\subsection{NE Analysis of the Non-cooperative Game} \label{sec:NEanalysis}
We consider this non-cooperative game Nash equilibrium (NE) as a solution to the users' nonce selection strategies. In this case, the NE is obtained by using the best response function which is the best strategy of one player given other users' strategies. The best response function of $i$-th user's nonce selection, given a vector of strategies offered by other users' nonce selection ${\bf{M}}_{-i}$, is defined as follows:
\begin{equation}\label{bestresponsefunc}
{R_i}\left( {{{\mathbf{M}}_{ - i}}} \right) = \arg \mathop {\max }\limits_{{M_i}} u\left( {{M_i},{{\mathbf{M}}_{ - i}}} \right),
\end{equation}
where ${{\mathbf{M}}_{ - i}} = \left( {{M_1},{M_2},...,{M_{i - 1}},{M_{i + 1}},...,{M_N}} \right)$ for $i=1,2,...,N$. The vector ${{\bf{M}}^*} = {\left( { \cdots {M_i}^* \cdots } \right)}$ denotes a NE of this game on nonce selection for 
\begin{equation}\label{NE}
M_i^* = {R_i}\left( {{\bf{M}}_{ - i}^*} \right),\;\;\;\;\forall i = 1,2,...,N,
\end{equation}
where ${{\bf{M}}_{ - i}^*}$ denotes the vector of best responses for player $j$ for $j \ne i$. The above variable $M_i$ has the non-negative integer constraint, making $u_i$ become a mixed integer programming function. Here, we remove the integer restriction and solve the optimal nonce selection strategy directly, and then get the integer solution by rounding down to nearest integer. We next analyze the existence of NE in the noncooperative game $\left( N, {\mathbf{M}}, {u_i} \left(  \cdot  \right) \right)$.

\begin{theorem}\label{ExistenceNE}
	A NE exists in the noncooperative game $\left( N, {\mathbf{M}}, {u_i} \left(  \cdot  \right) \right)$.
\end{theorem}

\begin{proof}\label{ProofExistenceNE}
	First of all, the strategy space for each user is defined to be $[0,{2^L-1}]$, which is a non-empty, convex, compact subset of the Euclidean space. From (\ref{user profit}), it is easier to treat discrete functions $u_i$ as continuous functions in $[0,{2^L-1}]$. Then we take the first order and second order derivatives of (\ref{user profit}) with $M_i$ to prove its concavity, which can be written as follows
	\begin{equation}\label{firstorder}
	\frac{{\partial {u_i}\left( {M_i} \right)}}{{\partial {M_i}}} = \left( {B + r{s_i}} \right){2^{-h}}\frac{{\sum {{M_{ - i}}} }}{{{{\left( {{M_i} + \sum {{M_{ - i}}} } \right)}^2}}} - c,
	\end{equation}
	\begin{equation}\label{secondorder}
	\frac{{\partial {u_i}^2\left( {M_i} \right)}}{{\partial {M_i}^2}} = \left( {B + r{s_i}} \right){2^{-h}}\frac{{ - 2\sum {{M_{ - i}}} }}{{{{\left( {{M_i} + \sum {{M_{ - i}}} } \right)}^3}}} < 0,
	\end{equation}
	for $i=1,2,...,N$, and where $\frac{{ - 2\sum {{M_{ - i}}} }}{{{{\left( {{M_i} + \sum {{M_{ - i}}} } \right)}^3}}} < 0$. Hence we have proved that $u_i$ is strictly concave with respect to $M_i$. Accordingly, the NE exists (see \cite{han2012game}-\textit{Theorem} 3.2) in the non-cooperative game $\left( N, {\mathbf{M}}, {u_i} \left(  \cdot  \right) \right)$. The proof is now completed. 
\end{proof}

Mathematically, to obtain the NE,  we have to address the following set of equations: 
\begin{equation}\label{solvesetequa1}
\frac{{\partial {u_i}\left( {\mathbf{M}} \right)}}{{\partial {M_i}}} = \left( {B + r{s_i}} \right){2^{-h}}\frac{{\sum {{M_{ - i}}} }}{{{{\left( {{M_i} + \sum {{M_{ - i}}} } \right)}^2}}} - c = 0,
\end{equation}
for all $i = 1,2,...,N$. Then we can obtain the best response function of user $i$ by solving (\ref{solvesetequa1})
\begin{equation}\label{solvesetequa2}
M_i^* = \sqrt {\frac{{\sum {{M_{ - i}}} \left( {B + r{s_i}} \right)}}{{c \cdot {2^h}}}}  - \sum {{M_{ - i}}},\;\text{for}\;i = 1,2,...,N.
\end{equation}
The solution $M_i^*$, which is a NE, can be obtained by solving the above set of linear equations by using a numerical method when all the parameters in (\ref{solvesetequa2}) are available.
We see other users' nonce selection strategies as an integrated part. Assume other parameters such as $B,r,s_i,c,h$ are fixed, then we get the first order and second order derivatives of (\ref{solvesetequa2}) with $\sum {{M_{ - i}}}$, which can be written as follows
\begin{equation}\label{Mifirstorder}
\frac{{\partial M_i^*}}{{\partial \left( {\sum {{M_{ - i}}} } \right)}} = \frac{1}{2}\sqrt {\frac{{\left( {B + r{s_i}} \right)}}{{c \cdot {2^h}}}} {\left( {\sum {{M_{ - i}}} } \right)^{ - \frac{1}{2}}} - 1,\;\text{for}\;i = 1,2,...,N,
\end{equation}
and
\begin{equation}\label{Misecondorder}
\frac{{{\partial ^2}M_i^*}}{{\partial {{\left( {\sum {{M_{ - i}}} } \right)}^2}}} =  - \frac{1}{4}\sqrt {\frac{{\left( {B + r{s_i}} \right)}}{{c \cdot {2^h}}}} {\left( {\sum {{M_{ - i}}} } \right)^{ - \frac{3}{2}}} < 0,\;\text{for}\;i = 1,2,...,N.
\end{equation}
Hence we easily obtain that the $i$-th user's best response $M_i^*$ is strictly concave with respect to other users' nonce selection strategies. In the first place, $M_i^*$ increases with the increment of $\sum {{M_{ - i}}}$, and then reaching an optimal point, it decreases when $\sum {{M_{ - i}}}$ increases. This conclusion is also consistent with the subsequent simulation result of Fig.~\ref{fig:bestresponse}. At last, combining the above analysis part, the optimal users' nonce selection under the untrusted MEC PoW scheme can be obtained, as summarized in Algorithm \ref{alternatingalgorithm1}.
\begin{algorithm}[!t]
	\caption{Alternating Optimization Algorithm of Users' Nonce Selection} \label{alternatingalgorithm1}
	\setcounter{algorithm}{2}
	\begin{algorithmic}[1]
		\STATE {\textbf{Initialization:} input data $\left( {B,r,{s_i},c,L,h} \right)$, set $i=1$ and choose ${M_{ - 1}} \in \left[ {0,{2^L-1}} \right]$.}
		\REPEAT 
		\STATE {Fix ${{\bf{M}}_{ - i}} = \left( {{M_1},{M_2},...,{M_{i - 1}},{M_{i + 1}},...,{M_N}} \right)$, and calculate $M_i^*$ using (\ref{solvesetequa2}).}
		\STATE {$i \leftarrow i + 1$}
		\STATE {Then update the set ${{\bf{M}}_{ - i}}=$$ ( {{M_1},{M_2},...,}$${M_{i - 1}^*},$ ${M_{i + 1}},...,$${M_N} )$.}
		\UNTIL {the optimal nonce vector  ${\bf{M}^*} = \left\{ {M_1^*,M_2^*,...,M_N^*} \right\}$ is obtained.}
	\end{algorithmic}
\end{algorithm}
From (\ref{solvesetequa2}),  we can see that the selected nonce numbers of all other users are known and it is apparently difficult to apply in the actual system. For this reason, the best response function can be rewritten in the analytical expression of \textit{Theorem}~\ref{finalNE}.
\begin{theorem}\label{finalNE}
	The NE for user $i$	in the noncooperative game (N,\textbf{M},${u_i}\left( . \right)$) is given by
	\begin{equation}\label{finalNE1}
	M_i^* = \frac{{N - 1}}{{\sum\limits_{j \in \mathcal{N}} {\frac{{c \cdot {2^h}}}{{B + r{s_j}}}} }} - {\left( {\frac{{N - 1}}{{\sum\limits_{j \in \mathcal{N}} {\frac{{c \cdot {2^h}}}{{B + r{s_j}}}} }}} \right)^2}\frac{{c \cdot {2^h}}}{{B + r{s_i}}},\;\text{for}\;i=1,2,...,N. 
	\end{equation}
\end{theorem}
\begin{proof}\label{prooffinalNE}
	According to (\ref{solvesetequa1}), for each user $i$, we have the mathematical expression
	\begin{equation}\label{prooffinalNE1}
	\frac{{\sum {{M_{ - i}}} }}{{{{\left( {{M_i} + \sum {{M_{ - i}}} } \right)}^2}}} = \frac{{c \cdot {2^h}}}{{B + r{s_i}}},\;\text{for}\;i=1,2,...,N.
	\end{equation}
	Then we calculate the summation of this expression (\ref{prooffinalNE1}) for all users as follows
	\begin{equation}\label{prooffinalNE2}
	\frac{{N - 1}}{{{M_i} + \sum {{M_{ - i}}} }} = \sum\limits_{j \in \mathcal{N}} {\frac{{c \cdot {2^h}}}{{B + r{s_j}}}},\;\text{for}\;i=1,2,...,N,
	\end{equation}
	we can obtain
	\begin{equation}\label{prooffinalNE3}
	{M_i} + \sum {{M_{ - i}}}  = \frac{{N - 1}}{{\sum\limits_{j \in \mathcal{N}} {\frac{{c \cdot {2^h}}}{{B + r{s_j}}}} }},\;\text{for}\;i=1,2,...,N.
	\end{equation}
	By substituting (\ref{prooffinalNE3}) into (\ref{solvesetequa2}), we can have
	\begin{equation}\label{prooffinalNE4}
	\frac{{N - 1}}{{\sum\limits_{j \in \mathcal{N}} {\frac{{c \cdot {2^h}}}{{B + r{s_j}}}} }} = \sqrt {\frac{{B + r{s_i}}}{{c \cdot {2^h}}}\left( {\frac{{N - 1}}{{\sum\limits_{j \in \mathcal{N}} {\frac{c \cdot {2^h}}{{B + r{s_j}}}} }} - {M_i}} \right)}, \;\text{for}\;i = 1,2,...,N.
	\end{equation}
	After squaring both sides and simple transformations, we can obtain the NE for user $i$ as shown in (\ref{finalNE1}).
\end{proof}
\begin{figure*}[!t]
	\centering 
	\subfigure[Fixed reward is changing]{ 
		\label{fig:subfig:c} 
		\includegraphics[width=0.478\linewidth]{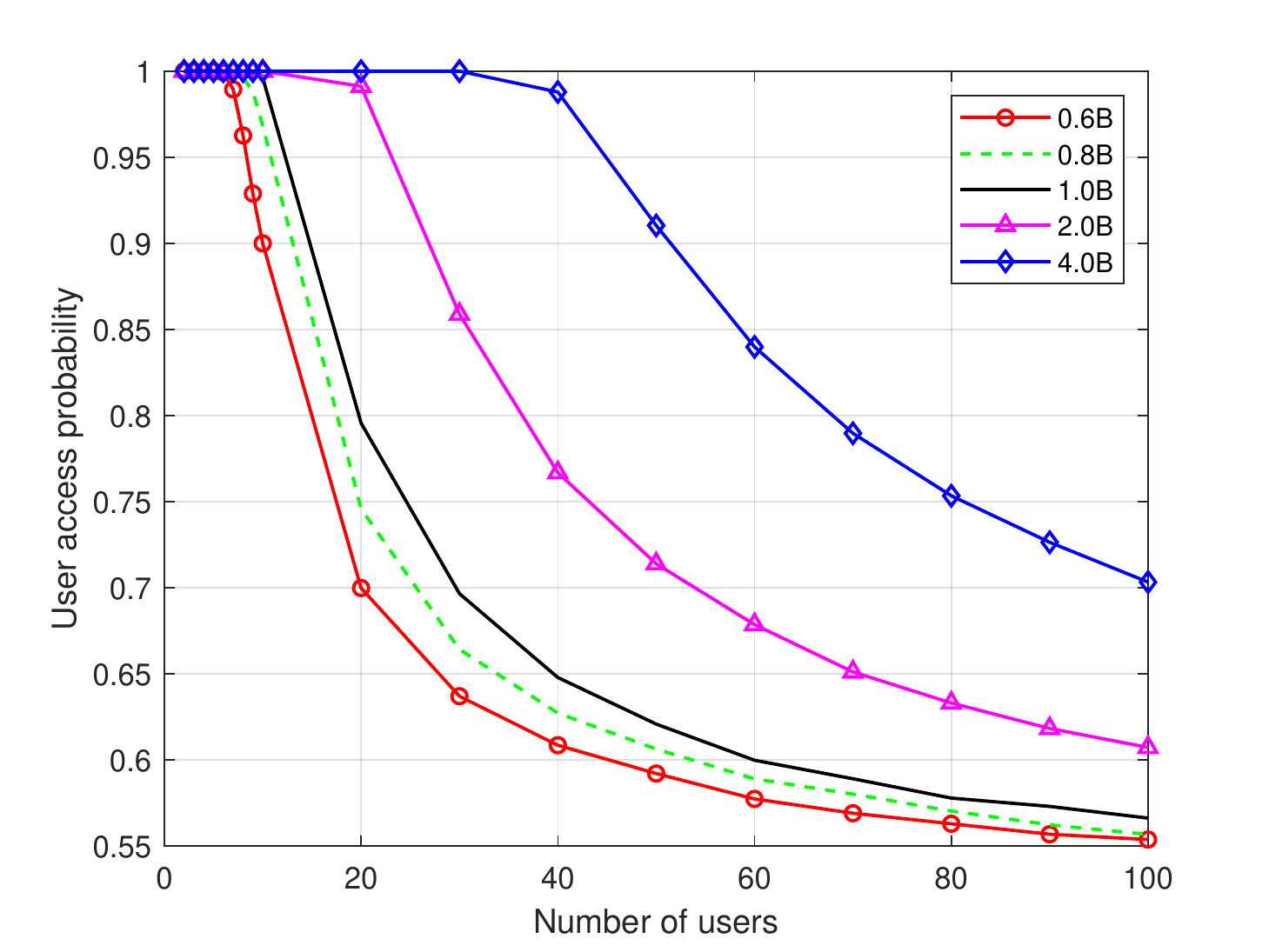}
	} 
	\subfigure[Difficulty factor is changing]{ 
		\label{fig:subfig:d} 
		\includegraphics[width=0.478\linewidth]{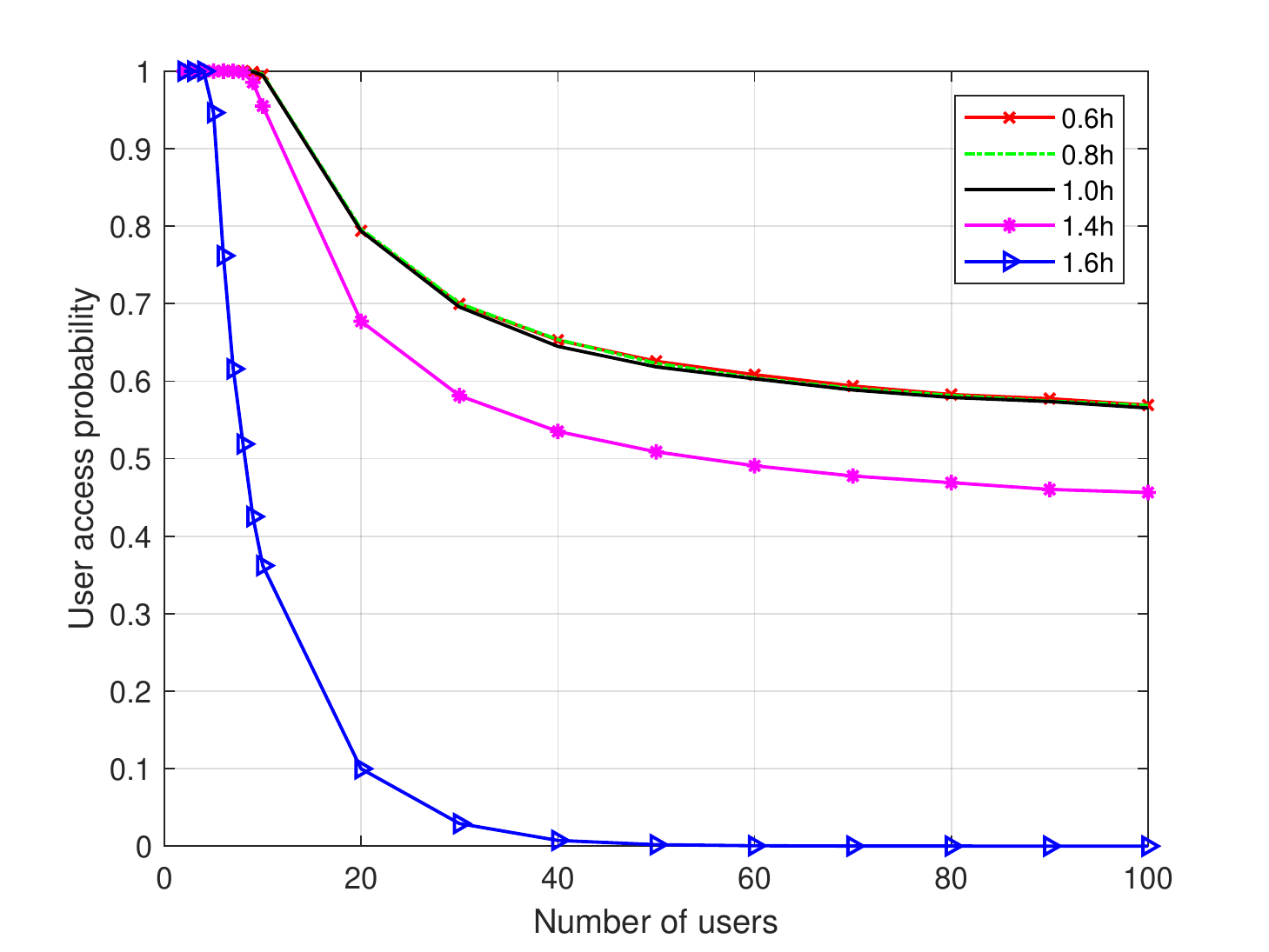}
	} 
	\caption{User access probability versus the number of users. $B=10^4$, $r=2$, $c=0.001$ and $h=12$.} 
	\label{fig:subfig2} 
\end{figure*}
\begin{remark}\label{RemarkFinalNE}
	(Properties of Theorem \ref{finalNE}): The analytical expression (\ref{finalNE1}) is distinct from (\ref{solvesetequa2}), which needs to know the global information about other users. But the analytical expression (\ref{finalNE1}) only has to know the partial information such as the block size and the number of other users. 
\end{remark}

For ease of analysis, the above $M_i$ is considered to be a continuous variable value. Whereas the $i$-th user's optimal nonce length $M_i^*$ is non-negative integer in real system. Then the formula (\ref{finalNE1}) is rounded down to obtain
\begin{equation}\label{RoundedDownM_i}
M_i^* = \left\lfloor {\frac{{N - 1}}{{\sum\limits_{j \in \mathcal{N}} {\frac{{c \cdot {2^h}}}{{B + r{s_j}}}} }} - {{\left( {\frac{{N - 1}}{{\sum\limits_{j \in \mathcal{N}} {\frac{{c \cdot {2^h}}}{{B + r{s_j}}}} }}} \right)}^2}\frac{{c \cdot {2^h}}}{{B + r{s_i}}}} \right\rfloor. 
\end{equation}

\begin{remark}\label{RemarkRoundedDownM_i}
	(Interpretation of formula (\ref{RoundedDownM_i})): From  (\ref{RoundedDownM_i}), $M_i^*$ is related to some parameters of this mobile blockchain system, such as $N$, $B$, $h$ and so on. From the perspective of analyzing $M_i^*$, the following cases will occur in the untrusted MEC PoW scheme.\\
	\textbf{Case 1:} When $M_i^* \le 0$ for no one user. We can see that $N$ users all participate in the mobile blockchain network for mining block, i.e., the number of users $N$ is unchanged, and the values of the system parameters are all appropriate in this case. \\
	\textbf{Case 2:} When $M_i^* \le 0$ for part users. Some users will quit the mobile blockchain system, because the system parameters are inappropriate. From the formula (\ref{RoundedDownM_i}), we can see that possible reasons result in quitting system for some users, for instance, the total number of user is too large, the mining revenue of user is too small, or the blockchain difficulty factor is too large. Once system parameters $N, b, r, h, c$ are fixed, the users having larger block sizes will be reserved, and the users having smaller block sizes will be exited. After some users exit the mobile blockchain system, we still need to continue to calculate the optimal nonce selection strategies for the remaining users. \\
	\textbf{Case 3:} When $M_i^* \le 0$ for all $N$ users. All users exit from the mobile blockchain system for mining, leading to the entire system crashing dramatically. The system parameters are unsuitable exceedingly, such as the total number of user too much large, the mining revenue of user too much small, the blockchain difficulty factor too much large. In this case, an incentive mechanism should be designed to activate the system.
\end{remark}

To demonstrate the effect of system parameters on user access for three cases, we compare the user access probabilities for different users by changing the fixed reward for successfully mining $B$ and the difficulty factor $h$ in Fig.~\ref{fig:subfig:c} and Fig.~\ref{fig:subfig:d} respectively. From Fig.~\ref{fig:subfig2}, we can see that when $N$ is small, the all users can access the mobile blockchain system to participate in mining block (\textbf{Case 1}); when $N$ is large, some users will quit our system (\textbf{Case 2 or Case 3}). The reasons are given as follows. As observed from Fig.~\ref{fig:subfig:c}, the user access probability decreases with the declining fixed reward $B$. This is because as $B$ reduces, each user can get less revenue, and thus some users decide to quit system. From Fig.~\ref{fig:subfig:d}, when $h$ increases, the user access probability decreases, for example, when $h$ increases by $60\%$ and $N>40$, the user access probability drops to $0$, i.e., the mobile blockchain system crashes dramatically. This is because as $h$ increases, the difficulty of mining block in the mobile blockchain network increases, and thus more users cannot deal with PoW puzzle and then quit system. When $N$ is large, or $B$ is small, or $h$ is large, the mobile blockchain system is highly competitive and less benefit, and thus the user access probability will be low and even the entire system will crash. 
	
In this paper, we mainly investigate the \textbf{Case 1}, where we analyze the optimal number of nonces selected by each user and the specific mathematical relationship with all system parameters when the system parameters are fixed. Some possible problems in the other two cases, such as how to dynamically adjust the system parameters, and how to design the incentive mechanism for solving the problem of system crashing, are not analyze theoretically in this paper. These may be important research topics in the future.
	
Here, we will consider a special case, where we assume to know the average block size of all mobile users and let $\bar s$ denotes the fixed average block size. Then we can rewrite the expression (\ref{RoundedDownM_i} ) simply as
\begin{equation}\label{FixedAverageBlockSize}
M_i^* = \left\lfloor {\frac{{N - 1}}{N}\frac{{B + r\bar s}}{{c \cdot {2^h}}}\left( {1 - \frac{{N - 1}}{N}} \right)} \right\rfloor,
\end{equation}
for $i=1,2,...,N$. From expression (\ref{FixedAverageBlockSize}), we can achieve that user's nonce selection strategy only depends on the total number of users $N$, when others parameters $B,r,c,h$ are fixed in the whole untrusted MEC-aided mobile blockchain network. Moreover, we can obtain that the $i$-th user's nonce selection strategy approaches to $0$ when $N$ is sufficiently large. In this case, all players are trying to select nonces as much as possible, severe competition for a large number of users often leads to low computing resource allocation. Since the untrusted MEC server coexists over a long period of time, the user's nonce selection game will be played for multiple times, where the unfair competition could be handled through mutual trust and cooperation. By repeating a game over multiple time, the players can be conscious of the previous behavior of the players and change their strategies accordingly. Next, we use the repeated game to boost cooperation among competitive players.
\subsection{Analysis by Using Repeated Game}\label{RepeatedGame}
Let ${\cal T}$ be number of stage of the game $\mathcal{G}$. The repeated game, denoted by $\mathcal{G}^{\cal T}$, consists of game $\mathcal{G}$ repeated for ${\cal T}+1$ time slots form $t=0$ until $t={\cal T}$. First of all,  we have to define the utility for a player $i$ in the finitely repeated game (FRG) (${\cal T} < \infty $) is expressed as
\begin{equation}\label{RepeatedGameUtility}
\begin{gathered}
U_i = \sum\limits_{t = 0}^{\cal T} {{\delta ^t}{u_i}\left[ t \right]}  \hfill \\
\;\;\;\;\; = \sum\limits_{t = 0}^{\cal T} {{\delta ^t}} \left( {\left( {B + r{s_i}} \right){2^{ - h}}\frac{{{M_i}\left( t \right)}}{{\sum\limits_{j \in \mathcal{N}} {{M_j}\left( t \right)} }} - c{M_i}\left( t \right)} \right) \; \forall i \in {\cal N}, \hfill \\ 
\end{gathered} 
\end{equation}
where $\delta  \in \left[ {0,1} \right)$ is a discount factor, and ${u_i}\left[ t \right]$ is player $i$'s payoff at the $t$-th time slot. When $\delta$ is closer to 1, the player is more patient.

\begin{theorem}\label{FRG}
	Consider FRG $\mathcal{G}^{\cal T}$, and the one-shot stage game $\mathcal{G}$ has a unique nonce selection strategy $\bf{M}^*$. Then $\mathcal{G}^{\cal T}$ has a unique subgame perfect equilibrium (SPE). In this unique SPE, there is ${\bf{M}}^t=\bf{M}^*$ for each $t=0,1,..,{\cal T}$ regardless of history.
\end{theorem}

\begin{proof}\label{proofFRG}
	By using backward induction, in the last stage ${\cal T}$, we will have ${\bf{M}}^{\cal T}=\bf{M}^*$ regardless of the history of the FRG. Then we can inversely move to the previous stage ${\cal T}-1$, the sub-game also has ${\bf{M}}^{{\cal T}-1}=\bf{M}^*$. With this argument and continuing inductively, we can know that in this repeated game $\mathcal{G}^{\cal T}$ of complete information, there exists ${\bf{M}}^t=\bf{M}^*$ for $t=0,1,..,{\cal T}$ regardless of history.
\end{proof}

Thus, FRG is impossible to bring cooperative behavior among players, because we expect the emergence of cooperative behavior in the infinitely repeated games (IRG), denoted by ${\mathcal{G}^\infty }$. Then we define the utility for a player $i$ in the IRG (${\cal T} = \infty $) is calculated as
\begin{equation}\label{IRG}
\begin{gathered}
{U_i} = \left( {1 - \delta } \right)\sum\limits_{t = 0}^{ + \infty } {{\delta ^t}{u_i}\left[ t \right]}  \hfill \\
\;\;\;\;\;= \left( {1 - \delta } \right)\sum\limits_{t = 0}^{ + \infty } {{\delta ^t}\left( {\left( {B + r{s_i}} \right){2^{ - h}}\frac{{{M_i}\left( t \right)}}{{\sum\limits_{j \in \mathcal{N}} {{M_j}\left( t \right)} }} - c{M_i}\left( t \right)} \right)},\; \forall i \in {\cal N},  \hfill \\ 
\end{gathered} 
\end{equation}
where the factor ${1 - \delta }$ is introduced as a normalization, to measure stage and repeated game utilities in the same units. In this IRG, we consider a cooperation behavior to maximize the total utility function for all players which is given by
\begin{equation}\label{TotalUtility}
\begin{gathered}
\mathop {\max }\limits_{{M_1},{M_2},...,{M_N}} {u^{\text{total}}} = \sum\limits_{i = 1}^N {\left( {\left( {B + r{s_i}} \right){2^{ - h}}\frac{{{M_i}}}{{\sum\limits_{j \in \mathcal{N}} {{M_j}} }} - c{M_i}} \right)}  \hfill \\
s.t.\;\;{M_i} > 0,\;\text{for}\;i = 1,2,...,N. \hfill \\ 
\end{gathered} 
\end{equation}
\begin{figure}[!t]
	\centering
	\includegraphics[width=0.6\linewidth]{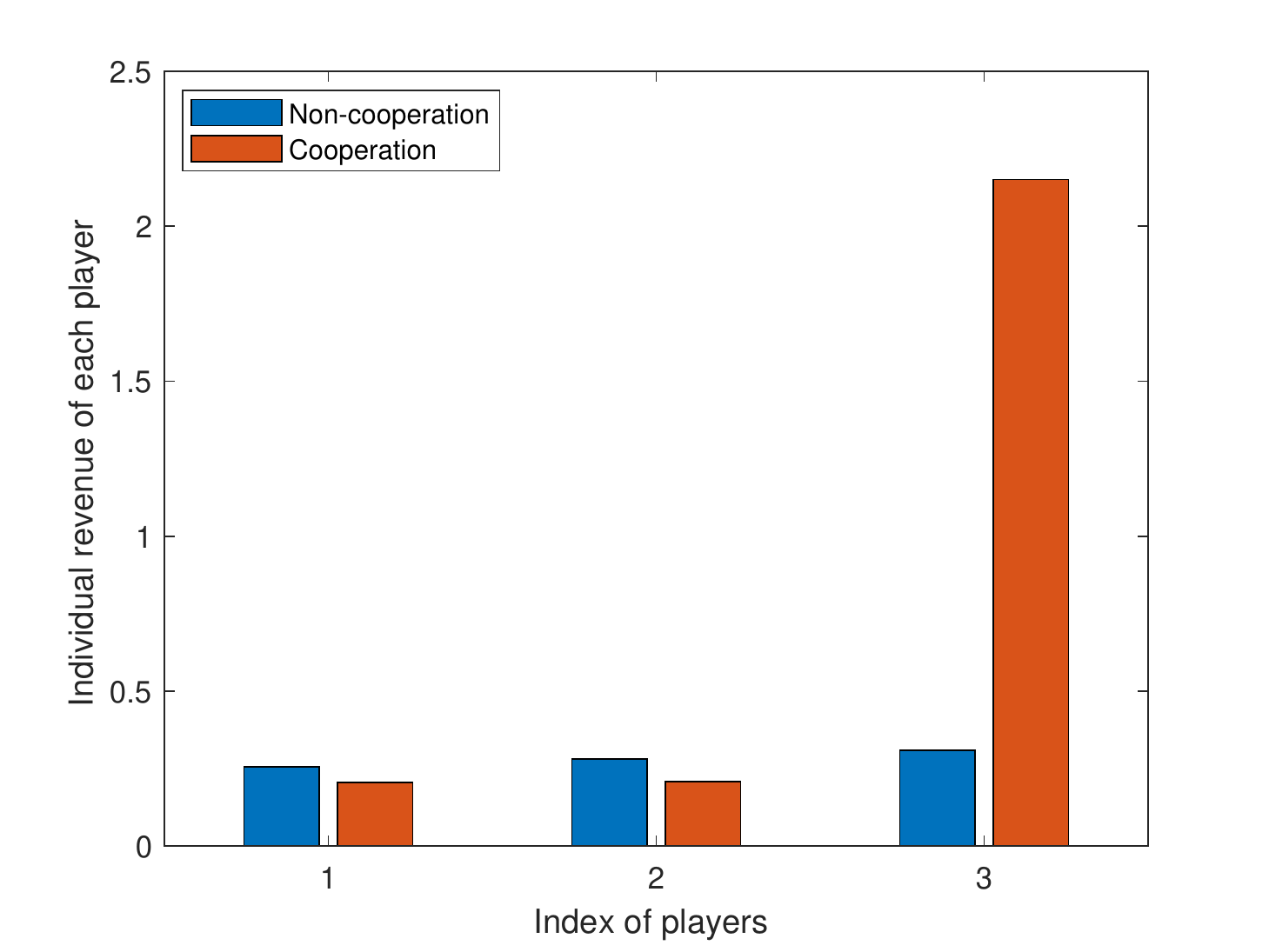}
	\caption{Comparison of individual revenue when players take non-cooperative and cooperative behaviors. $N=3$, $B = {10^4}$, $r = 2$, $c = 0.001$, $s_1=100$, $s_2=200$, $s_3=300$ and $h = 12$.}
	\label{fig:cooperationutility}
\end{figure}
The above optimization problem (\ref{TotalUtility}) is multiple variables non-convex problem. This optimization problem is more difficult to solve. Here, we only need simple analysis, therefore, there is no request to address this problem. Here, we simulate individual revenues of three player in the untrusted MEC-aided mobile blockchain network by using the cooperative or non-cooperative behaviors, as shown in Fig.~\ref{fig:cooperationutility}, where the parameter settings are consistent with the later simulation Section \ref{sec:SimulationResults}. Firstly, we can see that the total revenue of the cooperative behavior in IRG is much higher than that of the non-cooperative behavior in one-shot game. While the individual revenues of player 1 and player 2 according to the cooperation in repeated games do not achieve a higher revenue than using the one-shot non-cooperative game. Consequently, these players will not want cooperation with other players in the repeated game. By cooperation, the individual revenues of each player do not increase together, therefore, the repeated game is not appropriate to analyze user's nonce selection in the mobile blockchain network. Moreover, if the block size $s_i$ for all players is the same, i.e., $s_1=s_2=...=s_N$, the optimization problem (\ref{TotalUtility}) will be a linear problem. In this case, all user's nonce selection strategies are approximately equal to $0$, which does not accord with the actual situation. This further illustrates that cooperation approach of the repeated game is not feasible for IoT users in blockchain networks. 
\subsection{Blockchain's Difficulty Adjustment}\label{BloackchainDifficultyAdjustment}
Too high or too low difficulty can bring about different mining time. The blockchain network desires to ensure a stable average block time over long periods of time. For the sake of making the blockchain network run smoothly, the design of blockchain's difficulty adjustment is necessary. Next, we will focus more on blockchain's difficulty adjustment mechanism. We assume that the time for once nonce hash computing is $t_0$ and the total nonce number denotes as $M=M_1+M_2+...+M_N$ for a round of interaction between MEC server and mobile users. From the above analysis for user's nonce selection strategies in non-cooperative game $\mathcal{G}$, we can obtain the total expected optimal nonce number by
\begin{equation}\label{TotalNonce}
M=\frac{{N - 1}}{{\sum\limits_{j \in \mathcal{N}} {\frac{{c \cdot {2^h}}}{{B + r{s_j}}}} }}.	
\end{equation}
In one interaction, the longest time during MEC providing computing service to users process is $t_0M$. Each nonce hash computing is i.i.d. Bernoulli trial with the same successful probability, i.e., $p=P_D$, so the expected time from starting nonce hash computing to the successful mining can yield 
\begin{equation}\label{Mean}
\begin{gathered}
E\left( p \right) = \sum\limits_{k = 1}^{ + \infty } {{t_0}k{{\left( {1 - p} \right)}^{k - 1}}p}  \hfill \\
\;\;\;\;\;\;\;\;\; = \mathop {\lim }\limits_{K \to  + \infty } {t_0}\left( {\frac{{1 - {{\left( {1 - p} \right)}^K}}}{p} - K{{\left( {1 - p} \right)}^K}} \right) \hfill \\
\;\;\;\;\;\;\;\;\;= \frac{{{t_0}}}{p}, \hfill \\ 
\end{gathered} 
\end{equation}
where $k$ is the number of nonces, and $K$ denotes the upper bound of its number. Hence users may need to selected nonces according to multiple rounds of interaction between MEC server and users, so that block will be mined successfully. Let $\mathcal{R}$ is the interaction rounds of mining a new block between MEC server and users, and $\mathcal{R}$ should satisfy the following inequality,  
\begin{equation}\label{Rinequality}
{t_0}M\mathcal{R} \ge \frac{{{t_0}}}{p}, \hfill 
\end{equation}
Then we obtain
\begin{equation}\label{Rmin}
{\mathcal{R}_{\min }} = \frac{{\sum\limits_{j \in \mathcal{N}} {\frac{{c \cdot {4^h}}}{{B + r{s_j}}}} }}{{N - 1}}, 
\end{equation}
where $\mathcal{R}_{\min}$ is the minimum interaction rounds between MEC server and users for mining a new block. Here, we consider that a round time including the hash computing time and transmission time between MEC and users are almost the same.
It is assumed that a round time $\beta$ is fixed for the interaction process between users selecting nonces and MEC server providing computing service. Then the actual minimum time to mine a new block is 
\begin{equation}\label{ActualTime}
{T_{a}} = \beta {\mathcal{R}_{\min }}.
\end{equation}
So we keep the threshold number of rounds $\mathcal{R}_{th}$ fixed in the process of adjusting the difficulty, which is equivalent to keeping the block time unchanged.
Next, we are mainly interested in the average rounds for mining $G$ blocks, then we can obtain
\begin{equation}\label{AverageR}
\begin{gathered}
\bar {\mathcal{R}} = \frac{1}{G}\sum\limits_{g = 1}^G {{R_g}}   
= \frac{1}{G}\sum\limits_{g = 1}^G {\left( {\frac{{\sum\limits_{j \in \mathcal{N}} {\frac{{c \cdot {4^h}}}{{B + r{s_j}(g)}}} }}{{N - 1}}} \right)},  \hfill \\ 
\end{gathered} 
\end{equation}
where $\mathcal{R}_g$ denotes the interaction rounds between MEC server and all IoT devices for mining $g$-th new block. The variance of rounds for mining $G$ blocks is given as
\begin{equation}\label{VarianceR}
\begin{gathered}
\operatorname{var} \left( \mathcal{R} \right) = \frac{1}{G}\sum\limits_{g = 1}^G {\left( {\frac{1}{G}\sum\limits_{g = 1}^G {\left( {\frac{{\sum\limits_{j \in \mathcal{N}} {\frac{{c \cdot {4^h}}}{{B + r{s_j}(g)}}} }}{{N - 1}}} \right)} } \right.} - {\left. {\;\;\frac{{\sum\limits_{j \in \mathcal{N}} {\frac{{c \cdot {4^h}}}{{B + r{s_j}(g)}}} }}{{N - 1}}} \right)^2}. \hfill \\ 
\end{gathered} 
\end{equation}

From the above expression, we can see that the variance of rounds only depends on the number of mining blocks when other parameters are constant. Particularly, the variance of these rounds will  become worse with smaller block cycle $G$. Then, the rule of difficulty adjustment factor $h$ is to minimize the distance between the average rounds for $G$ blocks in (\ref{AverageR}) and the threshold rounds $\mathcal{R}_{th}$ as follows
\begin{equation}\label{Estimation_h}
\begin{gathered}
\hat h = \arg \mathop {\min }\limits_h \left| {\bar{\mathcal{R}} - {\mathcal{R}_{th}}} \right| \hfill \\
\;\;\; = \arg \mathop {\min }\limits_h \left| {\frac{1}{G}\sum\limits_{g = 1}^G {\left( {\frac{{\sum\limits_{j \in \mathcal{N}} {\frac{{c \cdot {4^h}}}{{B + r{s_j}(g)}}} }}{{N - 1}}} \right)} - {\mathcal{R}_{th}}} \right|. \hfill \\ 
\end{gathered} 
\end{equation}
When ${\bar {\mathcal{R}} = {\mathcal{R}_{th}}}$, then the optimal difficulty adjustment factor $h^*$ is expressed as follows
\begin{equation}\label{Optimal_h}
{h^*} = \frac{1}{2}{\log}\frac{{\left( {N - 1} \right)G{\mathcal{R}_{th}}}}{{\sum\limits_{g = 1}^G {\left( {\sum\limits_{j \in \mathcal{N}} {\frac{c}{{B + r{s_j(g)}}}} } \right)} }},
\end{equation}
where the $s_i(g)$ may change for mining different blocks. Eventually, combining the above analysis part, we update the difficulty factor $h$ based on previous $G$ blocks. 
We propose the blockchain's difficulty adjustment mechanism to be summarized as follows.

\begin{itemize}
	\item \textbf{Step 1:} Start with an initial difficulty factor $h_0$, yielding the initial rounds according to equation (\ref{Rmin}).
	
	\item \textbf{Step 2:} Keep $h_0$ constant over the difficulty adjustment interval of $G$ blocks. The average number of rounds can be obtained by (\ref{AverageR}).
	
	\item \textbf{Step 3:} After $G$ blocks, use adjustment rule (\ref{Estimation_h}) to update the difficulty factor $h$, then get a new $h^*$ according to (\ref{Optimal_h}).
	
	\item \textbf{Step 4:} Repeat this process with $h^*$ as initial difficulty factor.
\end{itemize}
It is easy to see that we will increase the difficulty factor $h$ if the average rounds of mining the previous $G$ blocks are shorter than $\mathcal{R}_{th}$, and decrease the difficulty factor $h$ if the average rounds of mining the previous $G$ blocks are longer than $\mathcal{R}_{th}$. According to this proposed difficulty adjustment mechanism, the blockchain network can keep block times stable during a long period of time.
\section{Numerical Results}\label{sec:SimulationResults}
In this section, we demonstrate simulation results to justify the effectiveness of our proposed nonce ordering algorithm in untrusted MEC PoW scheme, evaluate the expected optimal user's nonce selection strategies and the performance of blockchain's difficulty adjustment mechanism. Here, we consider a group of 3 users and each user's block size $s_i$ is uniformly distributed over $(0,1024]$, for $i=1,2,3$. The default parameter values are presented as follows: $B = {10^4}$, $r = 2$, $c = 0.001$, and $h = 12$. Note that some of these system parameters are varied in different simulation scenarios.
\subsection{Proposed Nonce Ordering Algorithm}\label{ProposedNonceOrderingAlgorithm}
First of all, in order to verify the rationality of the our proposed ordering algorithm, i.e., for any position in the merged sequence by using Algorithm \ref{orderingalgorithm1}, the number of selected nonces of each original sequence is proportional to its nonce length, we consider to compare the target and actual probability distribution. Without loss of generality, we assume only three users in the mobile blockchain system and submitted nonce numbers have a fixed proportion as $M_1:M_2:M_3=1:3:6$, i.e., the target probability distribution is $P=(0.1,0.3,0.6)$, while the total nonce numbers $M=M_1+M_2+M_3$ is changed. We compare our proposed ordering Algorithm 1 with weighted round robin (WRR) algorithm, where the weights of user 1, user 2 and user 3 are 1, 3 and 6, respectively. 
\begin{figure*}[!t]
	\centering 
	\subfigure[0.2M]{ 
		\label{fig:subfig:a} 
		\includegraphics[width=0.478\linewidth]{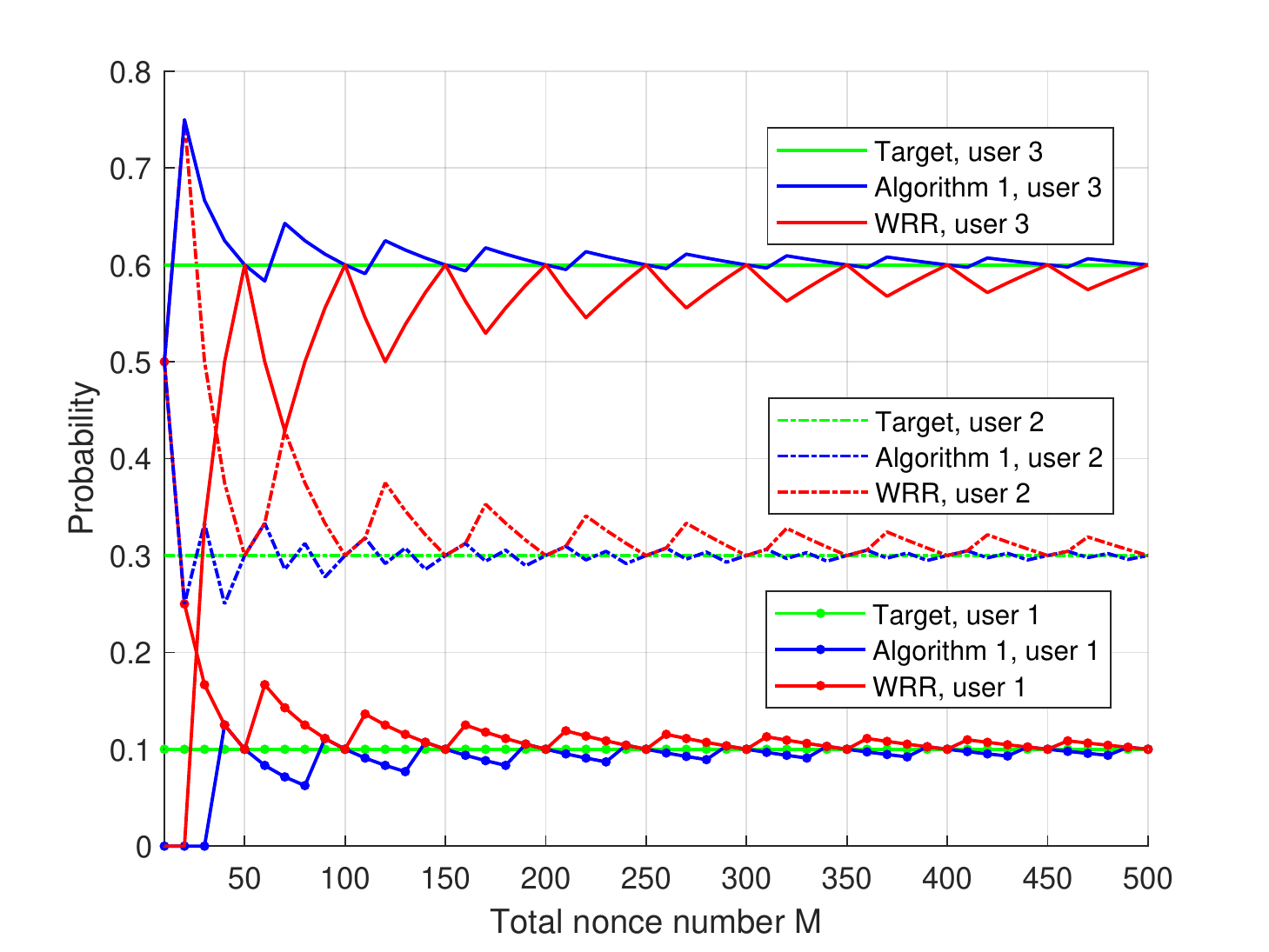}
	} 
	\subfigure[0.5M]{ 
		\label{fig:subfig:b} 
		\includegraphics[width=0.478\linewidth]{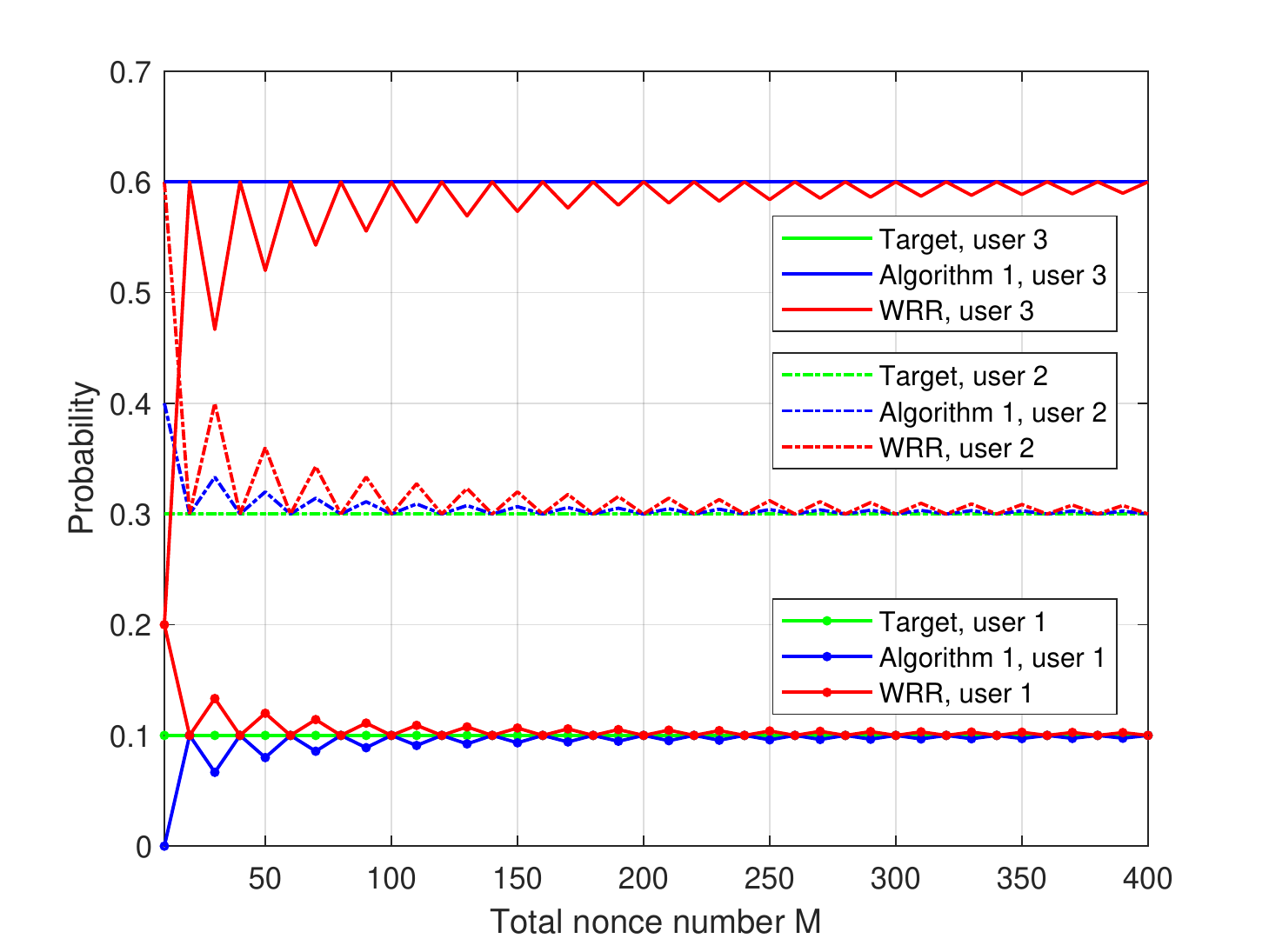}
	} 
	\caption{Comparison of target and actual probability distribution by Algorithm \ref{orderingalgorithm1} and WRR algorithm.} 
	\label{fig:subfig} 
\end{figure*}

Fig.~\ref{fig:subfig} reveals the comparison of target and actual probability distribution by Algorithm 1 and WRR algorithm. Here, we take the first $0.2M$ and $0.5M$ nonces of the merged sequence in Fig.~\ref{fig:subfig:a} and Fig.~\ref{fig:subfig:b}, respectively. By using Algorithm \ref{orderingalgorithm1}, it can be observed that actual probability has a little error with target probability when $M$ is small and then it approaches or even equals the target probability for all three users when $M$ increases in Fig.~\ref{fig:subfig:a} and Fig.~\ref{fig:subfig:b}. The reason why the performance of Algorithm \ref{orderingalgorithm1} is not good when the total nonce number $M$ is small, and then the nonce length of each user is smaller. For example, if $M=10$, user 1 has only one nonce and then we take any position of the merged sequence to verify the rationality of the ordering algorithm. The selected probability of user 1 mostly is $0$, not the target probability $0.1$. Moreover, because of the large amount of nonce computing in the  blockchain network (e.g., the SHA-256 hash function commonly used by Bitcoin requires $2^{32}$ nonce hash computing), the shortcoming of Algorithm \ref{orderingalgorithm1} can be ignored when $M$ is small. When the total nonce length approaches infinity, the user's successful mining probability is nearest to the target probability.

On the other hand, the fairness performance of the WRR algorithm is significantly worse than that of Algorithm \ref{orderingalgorithm1}, the WRR algorithm becomes better with the increment of $M$, however, the convergence of WRR algorithm is still poor due to the limitations of the WRR algorithm. For example, by using WRR algorithm, the user 3 with a high weight is always selected until the number of weights is reached, and then the next user 1 will be selected. This implies that the service is consecutively provided on the same user 3, resulting in the next user 1 with a low weight may be idle and having no smooth resource allocation. By comparing Fig.~\ref{fig:subfig:a} and Fig.~\ref{fig:subfig:b}, it can be see that our proposed nonce ordering algorithm, Algorithm \ref{orderingalgorithm1}, can provide a much fairer hash computing service for all users than the WRR algorithm.
\begin{figure}[!t]
	\begin{minipage}[t]{0.5\linewidth}
		\centering
		\includegraphics[width=1.0\linewidth]{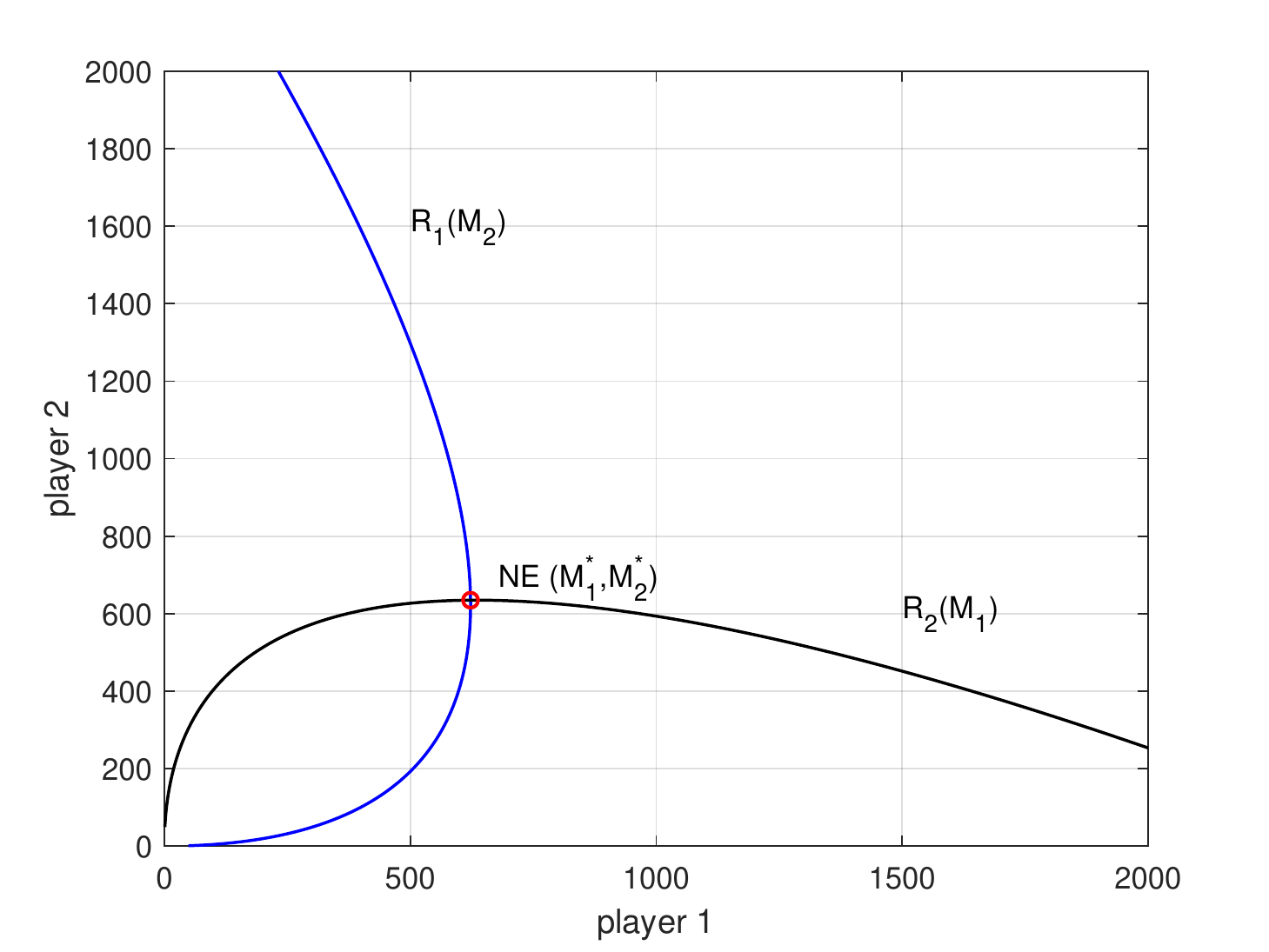}
		\caption{The players' best response functions\\ in the non-cooperative game.}
		\label{fig:bestresponse}
	\end{minipage}
	\begin{minipage}[t]{0.5\linewidth}
		\centering
		\includegraphics[width=1.0\linewidth]{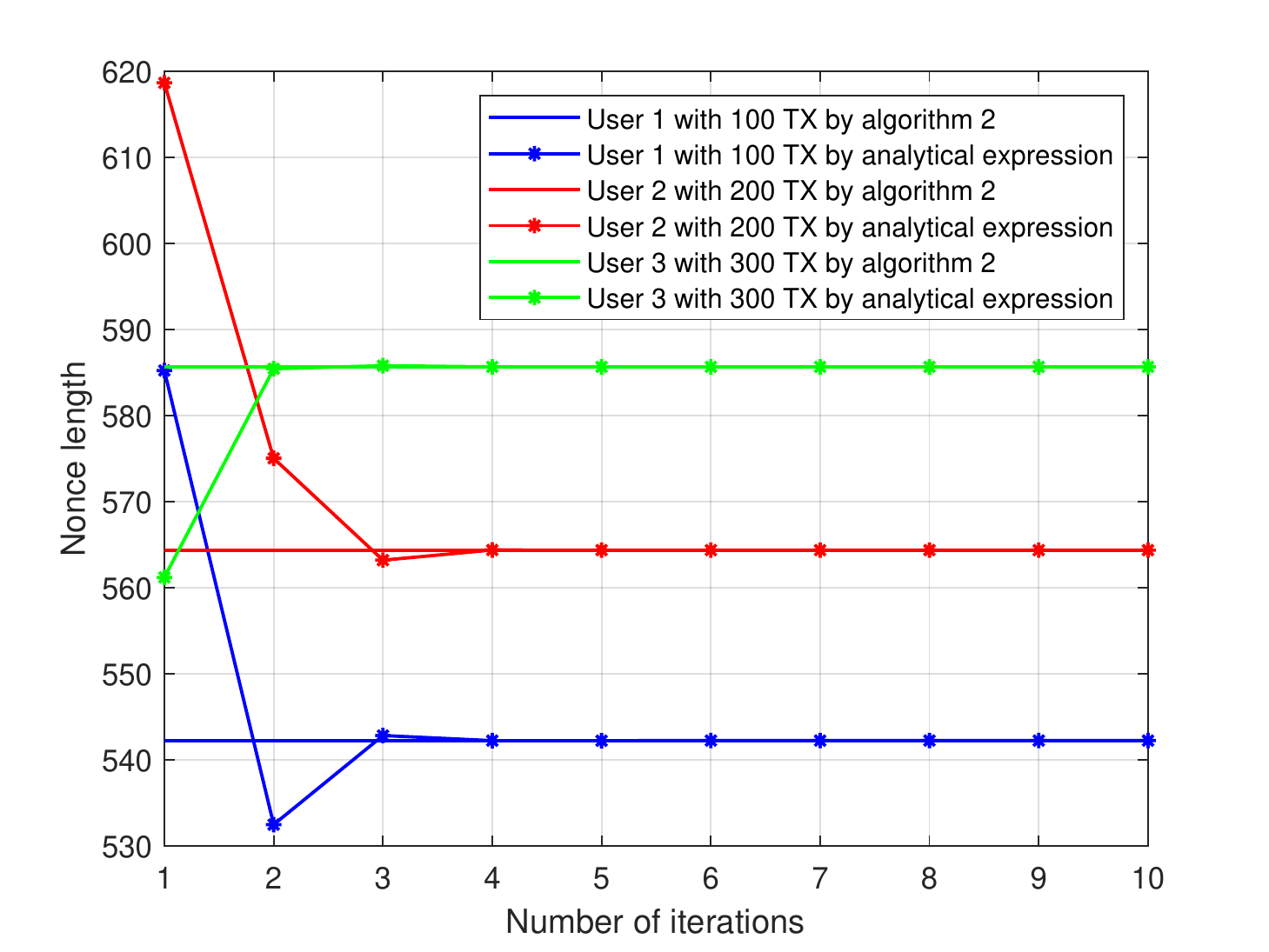}
		\caption{Nonce lengths for all users over\\ iterations.}
		\label{fig:alternatingalgorithm}
	\end{minipage}
\end{figure}
\subsection{Nonce Selection Strategies of Users }\label{NonceSelectionStrategiesofUsers}
We assume that only two users in this MEC-aided mobile blockchain network and the parameters of block size are $s_1=100$ and $s_2=200$. To find the NE we construct the best response functions by (\ref{NE}), for every player $i$, $i=1,2$, the best response functions can be obtained by using (\ref{solvesetequa2})
\begin{equation}\label{player1}
{R_1}\left( {{M_2}} \right) = \sqrt {\frac{{{M_2}\left( {B + r{s_1}} \right)}}{{c \cdot {2^h}}}}  - {M_2},
\end{equation} 
\begin{equation}\label{player2}
{R_2}\left( {{M_1}} \right) = \sqrt {\frac{{{M_1}\left( {B + r{s_2}} \right)}}{{c \cdot {2^h}}}}  - {M_1}.
\end{equation} 
In Fig.~\ref{fig:bestresponse}, we plot the best response functions of the players in (\ref{player1}) and (\ref{player2}). We can observe from Fig.~\ref{fig:bestresponse} that the best response function $R_1$ associates a unique strategy for player 1 to each strategy of player 2. Similarly, the best response function $R_2$ associates a unique strategy for player 2 to each strategy of player 1. As shown in Fig.~\ref{fig:bestresponse}, the two best response functions intersect at a unique point $(M_1^*,M_2^*)$. In fact, this point constitutes the unique pure-strategy NE of the non-cooperative game, at this point, every player's nonce selection strategy is the best response to the other player's strategies, i.e., $M_1^*=R_1(M_2^*)$ and $M_2^*=R_2(M_1^*)$. Thus we can conclude that this non-cooperative game exists a unique NE. Moreover, we can see that the player's nonce length decreases for remaining own revenue when other users' length is larger than their best strategy. These simulation results are consistent with the above NE analysis in Section \ref{sec:NEanalysis}. 
\begin{figure}[!t]
	\begin{minipage}[t]{0.5\linewidth}
		\centering
		\includegraphics[width=1.0\linewidth]{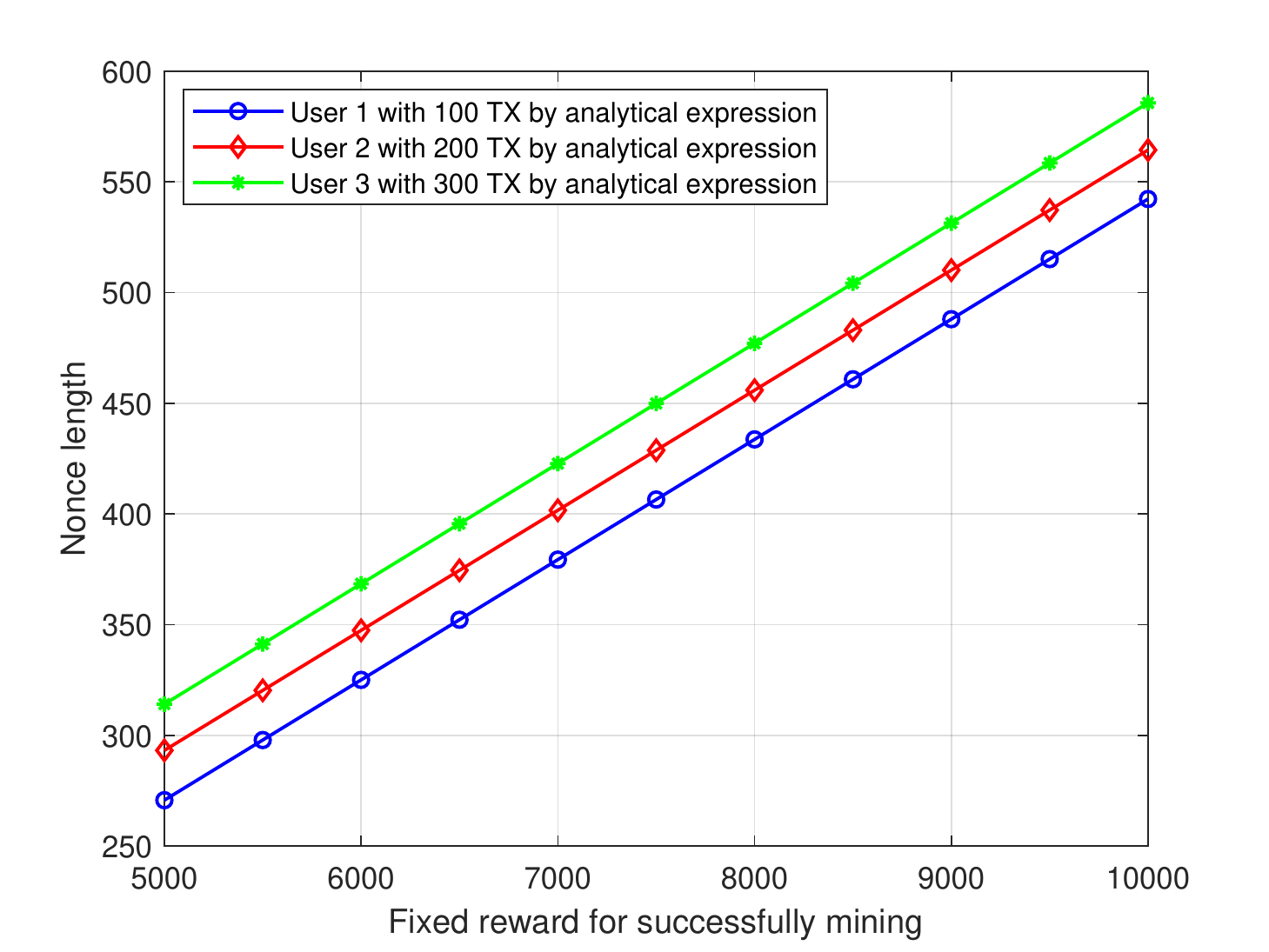}
		\caption{Nonce lengths for all users versus\\ the fixed reward for successfully mining.}
		\label{fig:noncelengthvsfixedreward}
	\end{minipage}%
	\begin{minipage}[t]{0.5\linewidth}
		\centering
		\includegraphics[width=1.0\linewidth]{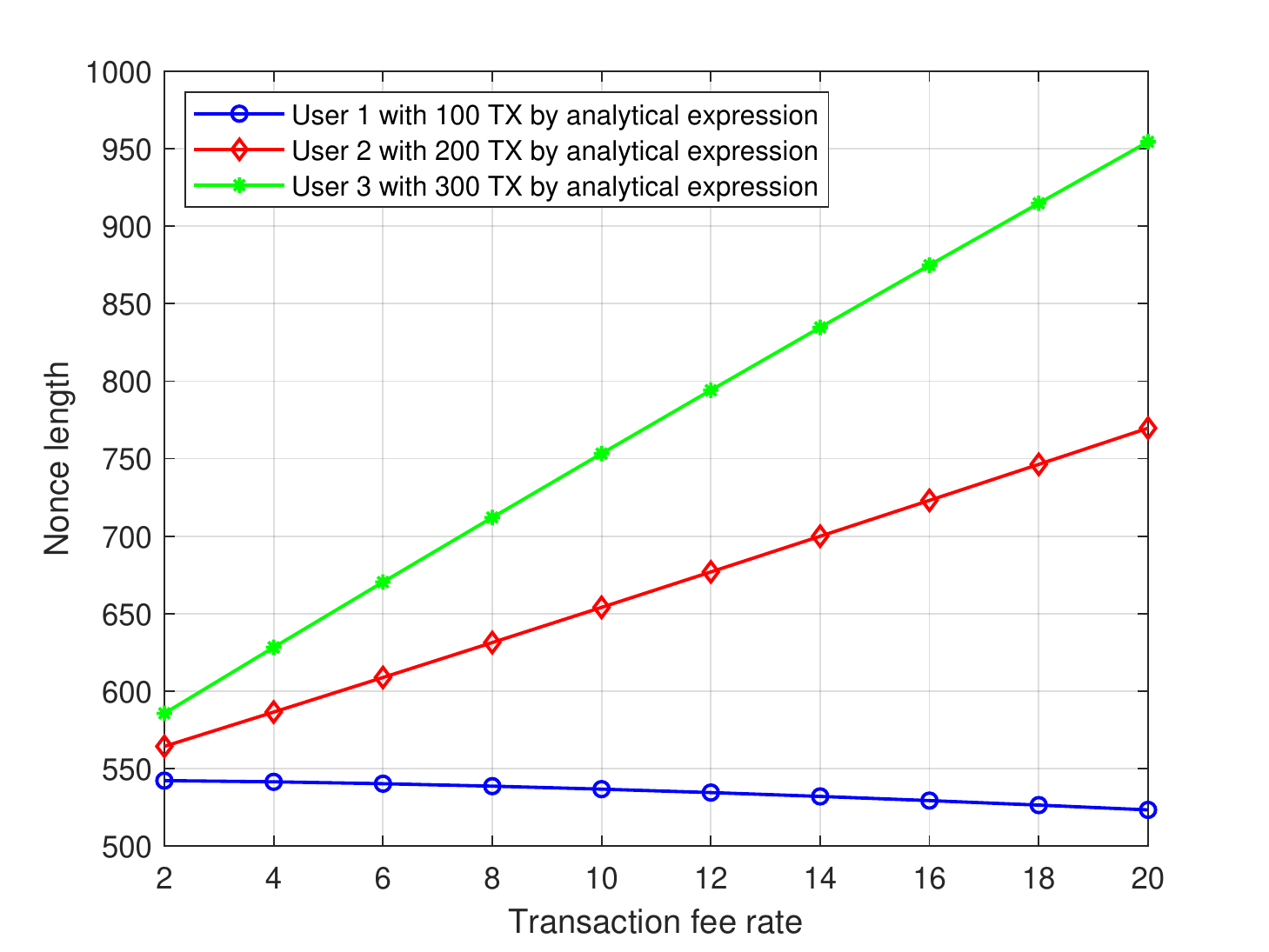}
		\caption{Nonce lengths versus the transaction\\ fee rate.}
		\label{fig:variablerewardfactor}
	\end{minipage}
\end{figure}

To demonstrate the performance of the proposed optimization algorithm of user's nonce selection and the impacts of various parameters on the nonce selection strategies, we consider three users in this MEC-aided mobile blockchain system with the following parameters of block size: $s_1=100,s_2=200,s_3=300$.
First, the convergence performance of the proposed alternating optimization algorithm of user's nonce selection is depicted in Fig.~\ref{fig:alternatingalgorithm}. We can see that nonce length of Algorithm \ref{alternatingalgorithm1} is unstable in the first 4 iterations. When the iteration is larger than 4, the alternating Algorithm \ref{alternatingalgorithm1} reaches a stable state and the gap between the proposed alternating Algorithm \ref{alternatingalgorithm1} and the analytical solution is rather small for all users, which verifies the good convergence of our proposed Algorithm \ref{alternatingalgorithm1} and the correctness of \textit{Theorem}~\ref{finalNE} about analytical expression (\ref{finalNE1}). The impacts of several parameters including fixed reward and transaction fee rate on user's nonce selection strategies are investigated in Fig.~\ref{fig:noncelengthvsfixedreward}-\ref{fig:variablerewardfactor}.
We illustrate nonce lengths for all users versus the fixed reward for successfully mining in Fig.~\ref{fig:noncelengthvsfixedreward}. We observed that the optimal nonce lengths for all users increase with the increment of fixed reward. Because of the fixed reward ascending, all users have greater incentives to cause higher nonce hash computing demand. Next, Fig.~\ref{fig:variablerewardfactor} illustrates nonce lengths for all users versus the transaction fee rate. We can see that the optimal nonce lengths for user 2 and user 3 raise with the transaction fee rate ascending. Nevertheless, the optimal nonce length for user 1 reduces with the increment of transaction fee rate. This reason is that the transaction fee rate increases, the incentive of each user are greater to have higher nonce hash computing demand. But the incentives of user 1 with smaller block size are still not much as that of user 2 and user 3, and becomes smaller than that of user 2 and user 3 when the transaction fee rate increases.
\begin{figure}[!t]
	\centering 
	\subfigure[$G=10$]{ 
		\label{fig:subfig:G=10} 
		\includegraphics[width=0.478\linewidth]{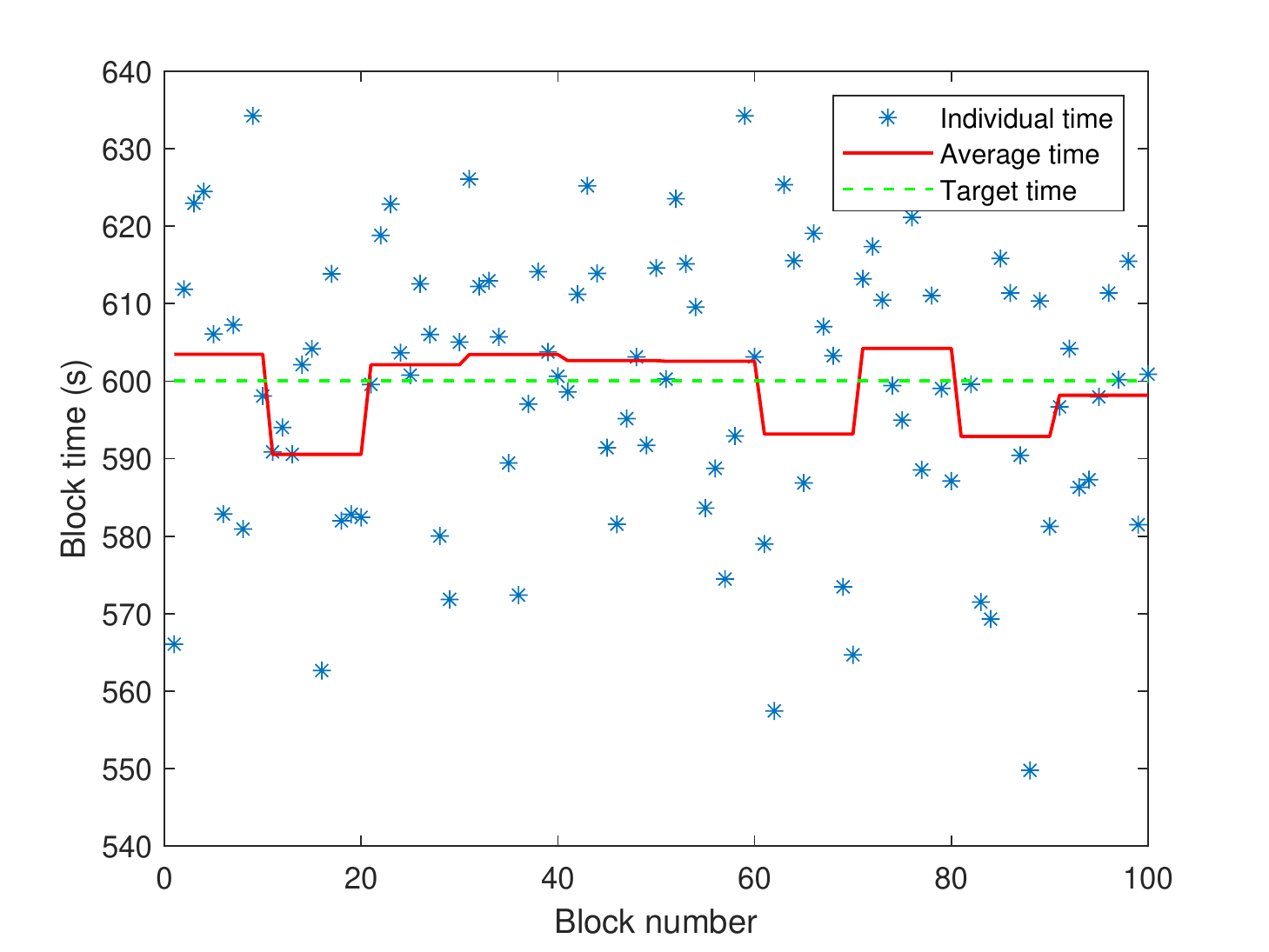}
	} 
	\subfigure[$G=2$]{ 
		\label{fig:subfig:G=2} 
		\includegraphics[width=0.478\linewidth]{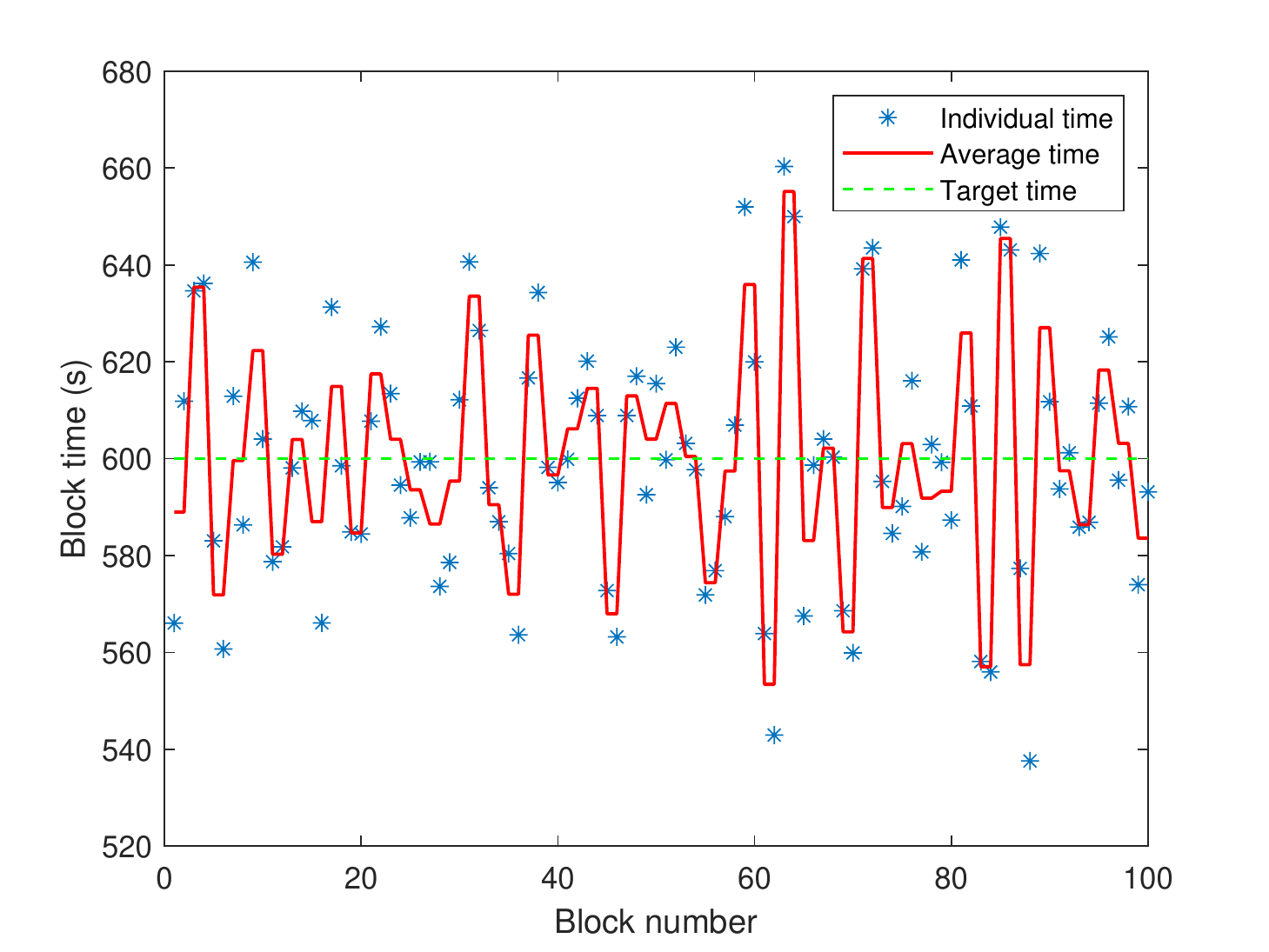}
	} 
	\caption{Block time versus the block number.} 
	\label{fig:blocktime} 
\end{figure}
\subsection{Difficulty's Adjustment Mechanism}
In the end, to display the performance of the proposed blockchain's difficulty adjustment mechanism, we plot the block time versus the block number for $G=10$ and $G=2$ in Fig.~\ref{fig:subfig:G=10} and Fig.~\ref{fig:subfig:G=2}, respectively. Here, we assume that each round time $\beta$ of the interaction process between MEC server and users is set to $120s$ and the target time of mining a block is denoted as $\beta R_{th}=600s$ for Bitcoin system. Note that blue '*' represents individual block times, and the red line presents the average block time within $G$ block intervals.
Firstly, in Fig.~\ref{fig:blocktime}, though the individual block time reveals a random property, the gap between the individual block time and the target time is within one minute, which ensures the effectiveness of our proposed difficulty adjustment mechanism. As observed from Fig.~\ref{fig:subfig:G=10} and Fig.~\ref{fig:subfig:G=2}, we can see that average block times fluctuate around the target time $600s$ every $G$ block intervals, because the blockchain's difficulty factor $h$ updates every $G$ blocks by using the proposed blockchain's difficulty adjustment mechanism. Meanwhile, the variance between the average block time and the target time of $G=2$ in Fig.~\ref{fig:subfig:G=2} is significantly larger than that of $G=10$ in Fig.~\ref{fig:subfig:G=10}, which is also consistent with the theoretical results of the above expression (\ref{VarianceR}). We can see that the variance of these block times for smaller $G$ becomes particularly worse than larger $G$. Accordingly, we should select the appropriate block intervals in the process of adjusting the blockchain's difficulty value.
\section{Conclusion}\label{sec:conclu}
In this paper, we have considered the untrusted MEC PoW scheme in the mobile blockchain-enabled IoT network. We have proposed a nonce ordering algorithm in this scheme, which guarantees that the number of selected nones of each original sequence is proportional to its nonce length for any position in the merged sequence. Then, the user's nonce selection strategies have been analyzed by using a non-cooperative game, and the NE has been considered as the solution. Furthermore, we have illustrated that cooperation approach of the repeated game is not feasible for IoT users in blockchain networks. Finally, we have designed the blockchain's difficulty adjustment mechanism to keep block times stable during a long period of time. Simulation results indicate that the nonce ordering algorithm can allocate computing resource of the untrusted MEC server  more fairly. The mobile IoT users can achieve optimal nonce selection strategies and we have verified that the proposed blockchain's difficulty adjustment mechanism is effective.
\bibliographystyle{IEEEtran}
\bibliography{IEEEabrv,BC_MEC}

\begin{thebibliography}{10}
\providecommand{\url}[1]{#1}
\csname url@samestyle\endcsname
\providecommand{\newblock}{\relax}
\providecommand{\bibinfo}[2]{#2}
\providecommand{\BIBentrySTDinterwordspacing}{\spaceskip=0pt\relax}
\providecommand{\BIBentryALTinterwordstretchfactor}{4}
\providecommand{\BIBentryALTinterwordspacing}{\spaceskip=\fontdimen2\font plus
\BIBentryALTinterwordstretchfactor\fontdimen3\font minus
  \fontdimen4\font\relax}
\providecommand{\BIBforeignlanguage}[2]{{%
\expandafter\ifx\csname l@#1\endcsname\relax
\typeout{** WARNING: IEEEtran.bst: No hyphenation pattern has been}%
\typeout{** loaded for the language `#1'. Using the pattern for}%
\typeout{** the default language instead.}%
\else
\language=\csname l@#1\endcsname
\fi
#2}}
\providecommand{\BIBdecl}{\relax}
\BIBdecl

\bibitem{nakamoto2008bitcoin}
S.~Nakamoto, ``Bitcoin: A peer-to-peer electronic cash system,'' Oct. 2008.

\bibitem{dai2019blockchain}
H.-N. {Dai}, Z.~{Zheng}, and Y.~{Zhang}, ``Blockchain for internet of things: A
  survey,'' \emph{arXiv:1906.00245}, Jun. 2019.

\bibitem{cao2019when}
B.~{Cao,} \emph{et~al.}, ``When internet of things meets blockchain: challenges
  in distributed consensus,'' \emph{arXiv:1905.06022}, May 2019.

\bibitem{conoscenti2016blockchain}
M.~{Conoscenti}, A.~{Vetrò}, and J.~C. {De Martin}, ``Blockchain for the
  internet of things: A systematic literature review,'' in \emph{IEEE/ACS 13th
  Int. Conf. of Comp. Syst. and Applica. (AICCSA)}, Nov. 2016, pp. 1--6.

\bibitem{banerjee2018blockchain}
M.~Banerjee, J.~Lee, and K.-K.~R. Choo, ``A blockchain future for internet of
  things security: A position paper,'' \emph{Digi. Commun. and Net.}, vol.~4,
  no.~3, pp. 149--160, Aug. 2018.

\bibitem{danzi2018analysis}
P.~Danzi, A.~E. Kalor, C.~Stefanovic{,} \emph{et~al.}, ``Analysis of the
  communication traffic for blockchain synchronization of {IoT} devices,'' in
  \emph{Proc. IEEE Int. Conf. Commun. (ICC)}, May 2018, pp. 1--7.

\bibitem{danzi2019delay}
P.~Danzi, A.~E. Kalor, C.~Stefanovic, and P.~Popovski, ``Delay and
  communication tradeoffs for blockchain systems with lightweight {IoT}
  clients,'' \emph{IEEE Internet Things J.}, vol.~6, no.~2, pp. 2354--2365,
  Apr. 2019.

\bibitem{xiong2018when}
Z.~Xiong, Y.~Zhang, D.~Niyato{,} \emph{et~al.}, ``When mobile blockchain meets
  edge computing,'' \emph{IEEE Commun. Mag.}, vol.~56, no.~8, pp. 33--39, Aug.
  2018.

\bibitem{xiong2018cloud}
Z.~Xiong{,} \emph{et~al.}, ``Cloud/fog computing resource management and
  pricing for blockchain networks,'' \emph{IEEE Internet Things J.}, pp. 1--1,
  Sept. 2018.

\bibitem{jiao2018social}
Y.~{Jiao}, P.~{Wang}, D.~{Niyato,} \emph{et~al.}, ``Social welfare maximization
  auction in edge computing resource allocation for mobile blockchain,'' in
  \emph{Proc. IEEE Int. Conf. Commun. (ICC)}, Jul. 2018, pp. 1--6.

\bibitem{jiao2019auction}
Y.~Jiao, P.~Wang, D.~Niyato, and K.~Suankaewmanee, ``Auction mechanisms in
  cloud/fog computing resource allocation for public blockchain networks,''
  \emph{IEEE Trans. Parallel Distrib. Syst.}, pp. 1--1, Mar. 2019.

\bibitem{luong2018optimal}
N.~C. Luong, Z.~Xiong, P.~Wang, and D.~Niyato, ``Optimal auction for edge
  computing resource management in mobile blockchain networks: A deep learning
  approach,'' in \emph{Proc. IEEE Int. Conf. Commun. (ICC)}, Jul. 2018, pp.
  1--6.

\bibitem{Bayhan2018Spass}
S.~{Bayhan}, A.~{Zubow}, and A.~{Wolisz}, ``Spass: Spectrum sensing as a
  service via smart contracts,'' in \emph{IEEE Int. Symp. on Dyna. Spec, Access
  Net. (DySPAN)}, Oct. 2018, pp. 1--10.

\bibitem{Kotobi2017blockchain}
K.~{Kotobi} and S.~G. {Bilén}, ``Blockchain-enabled spectrum access in
  cognitive radio networks,'' in \emph{Proc. Int. Wireless Commun. Symp.
  (WTS)}, Apr. 2017, pp. 1--6.

\bibitem{Gamal2019ASingle}
A.~{El Gamal} and H.~{El Gamal}, ``A single coin monetary mechanism for
  distributed cooperative interference management,'' \emph{IEEE Wireless
  Commun. Lett.}, vol.~8, no.~3, pp. 757--760, Jun. 2019.

\bibitem{garay2015bitcoin}
J.~Garay, A.~Kiayias, and N.~Leonardos, ``The bitcoin backbone protocol:
  Analysis and applications,'' in \emph{Annual Int. Conf. on the Theo. and
  Applica. of Crypto. Techn.}\hskip 1em plus 0.5em minus 0.4em\relax Springer,
  Apr. 2015, pp. 281--310.

\bibitem{bentov2016cryptocurrencies}
I.~Bentov, A.~Gabizon, and A.~Mizrahi, ``Cryptocurrencies without proof of
  work,'' in \emph{Int. Conf. on Finan. Crypto. and Data Secu.}\hskip 1em plus
  0.5em minus 0.4em\relax Springer, Aug. 2016, pp. 142--157.

\bibitem{mach2017mobile}
P.~{Mach} and Z.~{Becvar}, ``Mobile edge computing: A survey on architecture
  and computation offloading,'' \emph{IEEE Commun. Surveys Tuts.}, vol.~19,
  no.~3, pp. 1628--1656, Mar. 2017.

\bibitem{mao2017a}
Y.~Mao, C.~You, J.~Zhang{,} \emph{et~al.}, ``A survey on mobile edge computing:
  The communication perspective,'' \emph{IEEE Commun. Surveys Tuts.}, vol.~19,
  no.~4, pp. 2322--2358, Aug. 2017.

\bibitem{yu2018a}
W.~{Yu,} \emph{et~al.}, ``A survey on the edge computing for the internet of
  things,'' \emph{IEEE Access}, vol.~6, pp. 6900--6919, Nov. 2018.

\bibitem{liu2018computation}
M.~Liu, F.~R. Yu, Y.~Teng{,} \emph{et~al.}, ``Computation offloading and
  content caching in wireless blockchain networks with mobile edge computing,''
  \emph{IEEE Trans. Veh. Technol.}, vol.~67, no.~11, pp. 11\,008--11\,021, Aug.
  2018.

\bibitem{wu2018optimal}
Y.~{Wu,} \emph{et~al.}, ``Optimal computational power allocation in
  multi-access mobile edge computing for blockchain,'' \emph{Sensors}, vol.~18,
  no.~10, p. 3472, Oct. 2018.

\bibitem{liu2019distributed}
M.~{Liu}, F.~Richard{,} \emph{et~al.}, ``Distributed resource allocation in
  blockchain-based video streaming systems with mobile edge computing,''
  \emph{IEEE Trans. Wireless Commun.}, vol.~18, no.~1, pp. 695--708, Jan. 2019.

\bibitem{han2012game}
H.~Zhu, D.~Niyato, W.~Saad{,} \emph{et~al.}, \emph{Game Theory in Wireless and
  Communication Networks: Theory, Models, and Applications}.\hskip 1em plus
  0.5em minus 0.4em\relax Cambridge, UK: Cambridge university press, 2012.

\end{thebibliography}
\end{document}